\documentclass{article}
\usepackage{times,fullpage,amsthm,amsmath,amssymb,amsfonts,bbm}
\usepackage{graphicx, pstricks,pstricks-add,pst-3d}

\newenvironment{myitemize}
{\vspace{-1mm}\begin{list}{$\bullet$}%
    {\setlength{\itemsep}{-3pt}}%
    {\setlength{\topsep}{0pt}}%
    {\setlength{\partopsep}{0pt}}%
    {\setlength{\parsep}{0pt}}%
}
{\end{list}%
}

\newcommand{\floor} [1] {\lfloor #1 \rfloor}
\newcommand{\ket} [1] {\vert #1 \rangle}
\newcommand{\bra} [1] {\langle #1 \vert}
\newcommand{\braket}[2]{\langle #1 | #2 \rangle}

\newcommand{\smfrac}[2]{\mbox{$\frac{#1}{#2}$}}

\newtheorem{thm}{Theorem}
\newtheorem{lemma}{Lemma}
\newtheorem{defn}{Definition}
\newtheorem{cor}{Corollary}
\newtheorem{claim}{Claim}
\newtheorem{proposition}{Proposition}

\newcommand{\abs}[1]{\mid\! #1\! \mid}
\newcommand{\sgn}{\mathrm{sgn}}

\renewcommand{\L}{{\cal L}}
\newcommand{\C}{{\cal C}}
\newcommand{\Q}{{\cal Q}}

\newcommand{\Real}{\mathbb R}

\newcommand{\eps}{\varepsilon}

\newcommand{\A}{\mathcal{A}}
\newcommand{\B}{\mathcal{B}}
\newcommand{\V}{\mathcal{V}}
\newcommand{\norm}[1]{|\!|#1|\!|}

\newcommand{\Star}{{\displaystyle\boldsymbol\star}}

\newcommand{\pub}{\mathrm{pub}}
\newcommand{\ent}{\mathrm{ent}}

\newcommand{\id}{\mathbbm 1}

\newcommand{\conv}{\mathsf{conv}}
\newcommand{\aff}{\mathsf{aff}}

\newcommand{\aitch}{{\mathcal{H}}}
\newcommand{\ay}{\mathcal{A}}
\newcommand{\bee}{\mathcal{B}}
\newcommand{\ee}{\mathcal{E}}

\newcommand{\eks}{\mathcal{X}}
\newcommand{\wye}{\mathcal{Y}}

\newcommand{\cls}{\nu}
\newcommand{\qls}{\gamma_2}
\newcommand{\cmarg}{\nu}
\newcommand{\qmarg}{\gamma_2}
\newcommand{\cquasi}{\tilde{\nu}}
\newcommand{\qquasi}{\tilde{\gamma}_2}

\newcommand{\enlever}[1]{}

\newcommand{\remove}[1]{}

\usepackage[normalem]{ulem}	

\title{ The communication complexity of non-signaling distributions}
\author{Julien Degorre
	\thanks{Centre for Quantum Technologies, Singapore}
\and
Marc Kaplan
	\thanks{DIRO, Universit\'e de Montr\'eal}
\and
Sophie Laplante
	\thanks{LRI, Universit\'e Paris-Sud XI}
\and
J\'er\'emie Roland
	\thanks{NEC Laboratories America}
}
\date{}
\begin{document}

\maketitle

\begin{abstract}
We study a model of communication complexity that
encompasses many well-studied problems, including
classical and quantum communication complexity,
the complexity of simulating distributions arising from
bipartite measurements of shared quantum states, and
XOR games.  In this model, Alice gets an input $x$, Bob
gets an input $y$, and their goal is to each produce an output
$a,b$ distributed according to some pre-specified
joint distribution
$p(a,b|x,y)$.
Our results apply to any non-signaling distribution,
that is, those where Alice's marginal distribution does not
depend on Bob's input, and vice versa.

By taking a geometric view of the non-signaling distributions, we introduce
a simple new technique based on affine combinations of
lower-complexity distributions, and we give the first general technique
to apply to all these settings, with elementary proofs and
very intuitive interpretations.
Specifically, we introduce two complexity measures, one which
gives lower bounds on classical communication, and one for quantum communication.
These measures
can be expressed as convex optimization problems.
We show that the dual formulations have
a striking interpretation, since they
coincide with maximum violations of Bell and Tsirelson inequalities.
The dual expressions are closely related to the winning
probability of XOR games.
Despite their apparent simplicity, these lower bounds
subsume many known communication complexity
lower bound methods, most notably
the recent lower bounds of Linial and Shraibman for the
special case of Boolean functions.

We show that as in the case of Boolean functions, the gap
between the quantum and classical lower bounds is at most linear in the
size of the support of the distribution, and does not depend on the
size of the inputs.  This translates into a bound on
the gap between maximal Bell and Tsirelson inequality violations, which
was previously known only for the case of distributions with Boolean
outcomes and uniform marginals.  It also allows us to show that for
some distributions, information theoretic methods are necessary
to prove strong lower bounds.

Finally, we give an exponential upper bound on quantum and classical
communication complexity in the simultaneous messages model, for
any non-signaling distribution. One consequence of this is a simple proof that any quantum distribution can be approximated with a constant number of bits of communication.
\end{abstract}

{\section{Introduction}}

Communication complexity of Boolean functions has a long
and rich past, stemming from the paper of Yao in 1979~\cite{yao79},
whose motivation was to study the area of VLSI circuits.
In the years that followed, tremendous
progress has been made in developing a rich array of
lower bound techniques for various models of communication
complexity (see e.g.~\cite{KN97}).

From the physics side, the question of studying how much
communication is needed
to simulate distributions arising from physical phenomena,
such as measuring bipartite quantum states, was posed in 1992
by Maudlin, a philosopher of science, who wanted to quantify
the non-locality inherent to these systems~\cite{maudlin92}.
Maudlin, and the authors who followed~\cite{bct99,steiner00,tb03,cgmp05,dlr07}
(some independently of his work, and of each other)
progressively improved upper bounds on simulating
correlations of the 2 qubit singlet state.
In a recent breakthrough, Regev and Toner~\cite{rt07}
proved that two bits of communication
suffice to simulate the correlations arising from two-outcome
measurements of arbitrary-dimension bipartite quantum states.
In the more general case of non-binary outcomes, Shi and Zhu gave a protocol to approximate quantum distributions within constant error, using constant communication~\cite{shi05}. No non-trivial lower bounds are known for this problem.

In this paper, we consider the more general framework
of simulating non-signaling distributions.
These are distributions of the form $p(a,b|x,y)$,
where Alice gets input $x$ and produces an output $a$, and
Bob gets input $y$ and outputs $b$.
The non-signaling condition is a fundamental property of
bipartite physical systems, which
states that the players gain no information on the
other player's input.
In particular, distributions arising from
quantum measurements on shared bipartite states are
non-signaling, and Boolean functions
may be reduced  to extremal non-signaling distributions{ with Boolean outcomes and uniform marginals}.

Outside of the realm of Boolean functions, a very limited number of
tools are available to analyze the communication complexity of
distributed tasks, especially for quantum distributions with non-uniform
marginals.  In such cases, the distributions live in a larger-dimensional space and cannot be cast as communication matrices, so standard techniques do not apply.
The structure of non-signaling distributions has been the object of much study
in the quantum information community, yet outside the case of distributions
with Boolean inputs or outcomes~\cite{jones05,barrettpironio05}, or
with uniform marginal distributions,
much remains to be understood.

We introduce a new method to study
all non-signaling distributions, including the
case of non-Boolean outcomes and non-uniform marginals.
Our starting point is the observation that non-signaling distributions
coincide with affine (instead of convex) combinations of distributions that do
not require any communication, called local distributions.
With this elegant geometric formulation in mind, we show how to relate communication to
non-locality, where we measure non-locality by how far, in terms of its ``best'' affine representation,
a distribution is from the
convex set of local distributions.
Although they are formulated, and proven,
in quite a different way, our lower bounds turn out to subsume
Linial and Shraibman's nuclear and factorization norm lower bounds~\cite{ls07},
in the restricted case of Boolean functions.
Similarly, our upper bounds extend the upper bounds
of Shi and Zhu for
approximating quantum distributions~\cite{shi05}
to all non-signaling distributions (in particular distributions obtained by protocols using entanglement \emph{and} quantum communication).

Our complexity measures can be expressed as
convex optimization problems.
We may consider dual expressions, and
these turn out to correspond precisely to maximal Bell
inequality violations in the case of classical communication,
and Tsirelson inequality violations for quantum communication.
This confirms the long-held physics
intuition that large Bell inequality violations
should lead to large lower bounds on communication complexity.

We also show that there cannot be a large gap between the classical and quantum expressions.
This was previously known only in the case of distributions with
Boolean outcomes and
uniform marginals, and followed by Tsirelson's theorem and
Grothendieck's inequality, neither of which are known to
extend beyond this special case.  This also shows that our
method, as was already the case for Linial and Shraibman's
bounds, cannot hope to prove large gaps between classical and
quantum communication complexity.  While this is a negative
result, it also sheds some light on the relationship
between the Linial and Shraibman family of lower bound techniques,
and the information theoretic methods, such as the recent
subdistribution bound~\cite{jkn08}, one of the
few lower bound techniques not known to follow
from Linial and Shraibman.
We give an example of a
problem~\cite{bct99} for which rectangle size gives
an exponentially better lower bound than our method.

\vspace{2mm}

\paragraph{Summary of results}
The paper is organized as follows.  In Section~\ref{sec:preliminaries},
we give the required definitions and models of communication complexity
and characterizations of the classes of
distributions we consider.

In Section~\ref{sec:lower-bounds},
we prove our lower bound on classical and quantum
communication ({\bf Theorem~\ref{cor:lower-bound-quasi}}), and show that
it coincides with Linial and Shraibman's method in the special
case of Boolean functions
({\bf Theorems~\ref{lemma:cmargvscquasi} and~\ref{lem:epsilon-alpha}}).
Our lower bounds are convex optimization programs (linear programs in
the classical case), and in Section~\ref{sec:dual},
we show that the dual programs have a natural interpretation
in quantum information, as they coincide with Bell (or Tsirelson)
inequality violations ({\bf Theorem~\ref{thm:lp-bell}}).
We give a dual expression which also has a natural
interpretation, as the maximum winning probability of
an associated XOR game ({\bf Corollary~\ref{cor:LS-game}}).
The primal form turns out to be the multiplicative inverse
of the maximum winning probability of the associated XOR game,
where all inputs have the same winning probability.

In Section~\ref{sec:gamma-vs-nu}, we compare the two
methods and show that
the quantum and classical lower bound expressions can
differ by at most a factor that is linear in the number of outcomes
({\bf Theorem~\ref{thm:nu-gamma2}}).
When viewed as maximum Bell inequality violations, our
results imply that if Alice and Bob each have $k$ possible outcomes,
then the largest Bell inequality violation for quantum distributions is
at most $O(k^2)$.

Finally, in Section~\ref{sec:smp}, we give upper bounds
on simultaneous messages complexity in terms of our lower bound expression
({\bf Theorem~\ref{thm:smp}}).  We use
fingerprinting methods~\cite{bcwdw01, yaofinger03, shi05, gkr06}
to give very simple proofs that classical communication with
shared randomness, or quantum communication with shared entanglement,
can be simulated in the simultaneous messages model,
with exponential blowup in communication, and in particular that
any quantum distribution can be approximated with
constant communication.

\paragraph{Related work}

The use of affine combinations for non-signaling distributions
has roots in the quantum logic community, where quantum non-locality
has been studied within the setting of more general probability
theories~\cite{foulisrandall81, randallfoulis81, klayrandallfoulis87, wilce92}.
Until recently, this
line of work was largely unknown in the quantum information theory
community~\cite{barrett07, bblw}.

The structure of the non-signaling polytope has been the object
of much study.  A complete characterization of the vertices has
been obtained in some, but not all cases:
for two players, the case of binary inputs~\cite{barrettlinden05},
and the case of
binary outputs~\cite{barrettpironio05,jones05} are known, and for
$n$ players, the case of Boolean inputs and outputs is known~\cite{barrettpironio05}.

The work on simulating quantum distributions has focused mainly
on providing upper bounds, and most results apply to simulating
the correlations only.
In particular, Toner and Bacon show that projective measurements on a maximally entangled qubit pair may be simulated using one bit of communication~\cite{tb03}, and Regev and Toner extend this result by showing that the correlations arising from binary measurements on any entangled state may be simulated using two bits of communication only~\cite{rt07}.
A few results address the simulation of quantum distributions with
non-uniform marginals.  Bacon and Toner give an upper bound of 2
bits for non-maximally entangled
qubit pairs~\cite{tb03}.
Shi and Zhu~\cite{shi05} show a constant upper bound for approximating
any quantum distribution (including the marginals) to within a constant.

Pironio gives a general lower bound technique based on Bell-like
inequalities~\cite{pironio03}.
There are a few ad hoc lower bounds on
simulating quantum distributions, including a linear lower
bound for a distribution based on Deutsch-Jozsa's problem~\cite{bct99},
and a recent lower bound of Gavinsky~\cite{gavinsky09}.

The $\qls$ method was first introduced as a
measure of the complexity of matrices~\cite{lmss07}.
It was shown to be a lower bound on communication complexity~\cite{ls07},
and to generalize many previously known methods.
Lee \textit{et al.} use it to establish direct product theorems and
relate the dual norm of $\qls$
to the value of XOR games  \cite{lss08}. Lee and Shraibman~\cite{les07}
use a multidimensional generalization of
a related quantity $\mu$ (where the norm-1 ball consists of cylinder
intersections)
to prove a lower bound in the multiparty number-on-the-forehead-model,
for the disjointness function.


Since the first publication of this work, several extensions and improvements have been made to the upper bounds on Bell inequality violations of Section~\ref{sec:gamma-vs-nu}, and related lower bounds on the possible violations have been proved
~\cite{JPPVW09a,JPPVW10,JP10,BRSdW10}.

\section{Preliminaries}\label{sec:preliminaries}

In this paper, we extend the framework of communication complexity
to non-signaling distributions.  This framework encompasses
the standard models of communication complexity of Boolean functions
but also total and partial
non-Boolean functions and relations, as well as
distributions arising from the measurements of bipartite
quantum states.  Most results we present also
extend to the multipartite
setting.

\subsection{Definitions of the distribution classes}
Throughout this article, we consider bipartite conditional distributions $p(a,b|x,y)$
where $x\in \eks, y\in \wye$ are the inputs of the players,
and they are required to each produce an outcome
$a\in \ay, b\in \bee$, distributed according to $p(a,b|x,y)$. We will focus on so-called non-signaling distributions, where the marginal distribution of a given player's outcome does not depend on the other player's input. These include as a special case different classes of distributions, which we define in the following subsections.

\subsubsection{Local distributions}
In the quantum information literature, the distributions that
can be simulated with shared randomness and no communication
(also called a local hidden variable model) are called
local distributions.
\noindent
\begin{defn}\label{def:non-signaling}
{\em Local deterministic distributions}
are of the form $p(a,b|x,y)=\delta_{a=\lambda_A(x)} \cdot \delta_{b=\lambda_B(y)}$ where
$\lambda_A:\eks\rightarrow \ay$ and $\lambda_B:\wye \rightarrow \bee$,
and $\delta$ is the Kronecker delta.
A distribution is {\em local} if it can be written as
a convex combination of local deterministic distributions.
\end{defn}

We index by $\Lambda$ the set of local deterministic distributions
$\{{\mathbf p}^{\lambda}\}_{\lambda\in \Lambda}$ and denote
by $\L$ the set of local distributions.

\subsubsection{Quantum distributions}
Of particular interest in the study of quantum non-locality
are the distributions arising from measuring bipartite quantum states. We will use the following definition:
\noindent
\begin{defn}\label{def:quantum-distribution}
A distribution $\mathbf{p}$ is {\em quantum} if
there exists a bipartite quantum state $\ket{\psi}$ in a Hilbert space $\aitch=\aitch_A\otimes\aitch_B$ and
measurement operators $\{E_a(x):a\in \ay,x\in \eks\}$ acting on $\aitch_A$ and $\{E_b(y):b\in \bee, y\in \wye\}$ acting on $\aitch_B$,
such that $p(a,b|x,y)=\bra{\psi}E_a(x)\otimes E_b(y)\ket{\psi}$, with the measurement operators satisfying
\begin{enumerate}
 \item $E_a(x)^\dagger=E_a(x)$ and $E_b(y)^\dagger=E_b(y)$,
 \item $E_a(x)\cdot E_{a'}(x)=\delta_{aa'}E_a(x)$ and $E_b(y)\cdot E_{b'}(y)=\delta_{bb'}E_b(y)$,
 \item $\sum_aE_a(x)=\id_A$ and $\sum_bE_b(x)=\id_B$, where $\id_A$ and $\id_B$ are the identity operators on $\aitch_A$ and $\aitch_B$, respectively.
\end{enumerate}
\end{defn}


We denote by $\Q$ the set of all quantum distributions.

\subsubsection{Non-signaling distributions}

Non-signaling, a fundamental postulate of physics, states that no
observation on part of a system can instantaneously
affect a remote part of the system, or similarly, that
no signal can travel instantaneously.
For a bipartite probability distribution $p(a,b|x,y)$ describing observations on two distant physical systems, this means that no choice of measurement $y$ on Bob's side can affect the marginal distribution of the observed outcome $a$ on Alice's side, and vice versa.
Mathematically, non-signaling (also called causality)
is defined as follows.

\noindent
\begin{defn}[Non-signaling distributions]
A bipartite, conditional distribution $\mathbf p$ is non-signaling if
\begin{eqnarray*}
	\forall a,x,y,y',& \sum_b p(a,b|x,y) = \sum_b p(a,b|x,y') ,\\
	\forall b,x, x',y,& \sum_a p(a,b|x,y) = \sum_a p(a,b|x',y).
\end{eqnarray*}
\end{defn}

For any non-signaling distribution, the marginal
distribution on Alice's output
$p(a|x,y) = \sum_b p(a,b|x,y)$ does not
depend on $y$, so we write $p(a|x)$, and similarly
$p(b|y)$ for the marginal distribution on Bob's output.
We denote by $\C$ the set of all non-signaling distributions.

In the case of binary outcomes, that is, $\ay=\bee=\{\pm1\}$,
it is known that a non-signaling distribution is uniquely determined
by the (expected)
correlations, defined as $C(x,y)= E(a\cdot b|x,y)$, and the (expected)
marginals, defined as $M_A(x)=E(a|x), M_B(y)=E(b|y)$.

\noindent
\begin{proposition}\label{prop:representation}
For any functions $C:\eks\times \wye \rightarrow [-1,1] $,
$M_A:\eks\rightarrow  [-1,1]$,
$M_B:\wye\rightarrow [-1,1]$, satisfying
$1+ a\cdot b\; C(x,y) + a M_A(x) + b M_B(y)\geq 0$ $\forall (x,y)\in \eks\times \wye$ and $a,b\in\{\pm1\}$,
there is a unique non-signaling distribution
$\mathbf{p}$ such that $\forall\ x,y, E(a\cdot b|x,y)=C(x,y)$ and $E(a|x)=M_A(x)$
and $E(b|y)=M_B(y)$, where $a,b$ are distributed according to $\mathbf{p}$.
\end{proposition}
\begin{proof}
Fix $x,y$.  $C, M_A, M_B$ are obtained from $\mathbf p$
by the following full rank system of equations.
$$ \left( \begin{array}{rrrr}
		1 & -1 & -1  & 1\\
		1 & 1 & -1  & -1\\
		1 & -1 & 1  & -1\\
		1 & 1 & 1  & 1
	\end{array}
 \right)
 \left( \begin{array}{r}
		p(+1,+1|x,y)\\
		p(+1,-1|x,y)\\
		p(-1,+1|x,y)\\
		p(-1,-1|x,y)
	\end{array}
 \right)
= \left( \begin{array}{c}
		C(x,y)\\
		M_{A}(x)\\
		M_{B}(y)\\
		1
	\end{array}
 \right).
$$
Computing the inverse yields
$p(a,b|x,y) = \frac{1}{4}(1+ a\cdot b\; C(x,y) + a M_A(x) + b M_B(y))$.
\end{proof}

We will write ${\mathbf p}=(C, M_A, M_B)$ and use both notations interchangeably when considering distributions over binary outcomes.
We also denote by $\C_0$ the set of non-signaling distributions
with uniform marginals, that is, ${\mathbf p}=(C, 0, 0)$,
and write $C\in \C_0$, omitting the marginals when
there is no ambiguity.


Since local and quantum distributions are non-signaling, we use similar
notation for local and quantum distributions where
binary outcomes are concerned.  In the case of local distributions,
since the vertices of the polytope are deterministic strategies,
correlations and marginals can be written using $\pm1$ vectors.
Let $\conv(A)$ denote the convex hull of $A$.

\begin{proposition}
$\L = \conv(\{ (u^Tv,u,v) : u\in \{\pm1\}^\eks, v \in \{\pm1\}^\wye\})$.
\end{proposition}
We also denote by~$\L_0$ the set of
local correlations over binary outcomes with uniform marginals
and we let $\Q_0$ be the
set of all quantum correlations.

\subsubsection{Boolean functions}

There is a natural way to map a Boolean function $f: \eks \times \wye \rightarrow \{\pm1\}$ to a non-signaling distribution $p_f(a,b|x,y)$ over binary outcomes $a,b\in\{\pm1\}$, as follows:
\begin{defn}\label{def:boolean-functions}
For a function $f:\eks\times \wye \rightarrow \{-1, 1\}$,
denote $\mathbf{p}_f$ the distribution defined by
$p_f(a,b|x,y) = \frac{1}{2}$ if $f(x,y)=a\cdot b$ and
0 otherwise. Equivalently, $\mathbf{p}_f=(C_f,0,0)$ where $C_f(x,y)=f(x,y)$.
\end{defn}
By stipulating that the product of the players' outputs equals the value of the function, we see that the distribution has the same communication complexity as the function (up to an additional bit of communication, for Bob
to output $f(x,y)$).
As we shall see in Section~\ref{subsec:characterization-classical},
it so happens that the distributions associated to Boolean functions are extremal
points of the non-signaling polytope.

In the case of randomized communication complexity, a protocol
that simulates a Boolean function with error probability $\epsilon$
corresponds to simulating correlations $C'$ scaled down by a factor at most
$1-2\epsilon$, that is, $\forall x,y, \sgn(C'(x,y))=C_f(x,y)$ and
$\abs{C'(x,y)} \geq 1-2\epsilon$.  While we will not consider these
cases in full detail, non-Boolean functions, partial functions and some
classes of relations may be handled in a similar fashion, hence
our techniques can be used to show lower bounds in these settings
as well.

\subsection{Characterizations and relations among the distribution classes}

\subsubsection{Non-signaling distributions}
\label{subsec:characterization-classical}

The quantum information literature reveals a great deal of
insight into the structure of the
classical, quantum, and non-signaling distributions.
It is well known that $\L$ and $\C$ are polytopes.
While the extremal points of $\L$ are simply the local deterministic distributions, the non-signaling polytope $\C$ has a more complex structure~\cite{jones05, barrettpironio05}.
 In the case of $\C_0$, it is the convex hull of the distributions
obtained from Boolean functions.
\begin{proposition}
$\C_0 = \conv(\{ (C_f,0,0) : C_f\in \{\pm 1\}^{\eks\times\wye} \})$.
\end{proposition}

We show that $\C$ is the affine hull of the local polytope
(restricted to the positive orthant since all probabilities $p(a,b|x,y)$ must be positive). We give a simple proof for the case of binary outcomes but this carries over to the general case.
This was shown independently of us, on a few occasions
in different communities~\cite{randallfoulis81, foulisrandall81,
klayrandallfoulis87, wilce92, barrett07}.
\begin{thm}
\label{thm:qlhv}
$\C = \aff^+(\L)$,
where $\aff^+(\L)$ is the restriction to the positive orthant of the affine hull of $\L$, and $\dim\C=\dim\L=|\eks|\times|\wye|+|\eks|+|\wye|$.
\end{thm}
\begin{proof}
 We show that $\aff(\C)=\aff(\L)$. The theorem then follows by restricting to the positive orthant, and using the fact that $\C=\aff^+(\C)$.

[{\bf$\aff(\L)\subseteq \aff(\C)$}] Since any local distribution satisfies the (linear) non-signaling constraints in Def.~\ref{def:non-signaling}, this is also true for any affine combination of local distributions.

[{\bf$\aff(\C)\subseteq \aff(\L)$}] For any $(\sigma,\pi)\in \eks\times \wye$, we define the distribution $\mathbf{p}_{\sigma\pi}=(C_{\sigma\pi},u_{\sigma\pi},v_{\sigma\pi})$ with correlations $C_{\sigma\pi}(x,y)=\delta_{x=\sigma}\delta_{y=\pi}$ and marginals $u_{\sigma\pi}(x)=0,v_{\sigma\pi}(y)=0$.
Similarly, we define for any $\sigma\in \eks$ the distribution $\mathbf{p}_{\sigma\cdot}=(C_{\sigma\cdot},u_{\sigma\cdot},v_{\sigma\cdot})$ with $C_{\sigma\cdot}(x,y)=0,u_{\sigma\cdot}(x)=\delta_{x=\sigma},v_{\sigma\cdot}(y)=0$, and for any $\pi\in \wye$ the distribution $\mathbf{p}_{\cdot\pi}=(C_{\cdot\pi},u_{\cdot\pi},v_{\cdot\pi})$ with $C_{\cdot\pi}(x,y)=0,u_{\cdot\pi}(x)=0,v_{\cdot\pi}(y)=\delta_{y=\pi}$. It is straightforward to check that these $|\eks|\times|\wye|+|\eks|+|\wye|$ distributions are local, and that they constitute a basis for the vector space embedding $\aff(\C)$, which consists of vectors of the form $(C,u,v)$.
\end{proof}

This implies that while local distributions are \emph{convex} combinations of local deterministic distributions ${\mathbf p}^\lambda\in\Lambda$, non-signaling distributions are \emph{affine} combinations of these distributions.
\begin{cor}[Affine model]
A distribution ${\mathbf p} {\in}\C$ if and only if
$\,\exists q_\lambda\in\Real $ with ${\mathbf p} =\sum_{\lambda \in \Lambda} q_\lambda {\mathbf p} ^\lambda$.
\end{cor}
Note that since ${\mathbf p}$ is a distribution,
this implies $\sum_{\lambda \in \Lambda} q_\lambda =1$.
Since weights in an affine combination may be negative, but still
sum up to one, this may be interpreted as a \emph{quasi-mixture}
of local distributions, some distributions being used with
possibly ``negative probability''.
Surprisingly this is not a new notion;
see for example Groenewold~\cite{groenewold85} who gave an
affine model for quantum distributions; or a discussion of
``negative probability'' by Feynman~\cite{feynman86}.

\subsubsection{Quantum distributions}
\label{subsec:characterization-quantum}

The following fundamental theorem of Tsirelson relates measurements on quantum states to
the inner product of vectors.
\begin{thm}[\cite{tsi85}]\label{thm:tsirelson} Let $\mathbb{S}_n$ be the set of unit vectors in $\mathbb{R}^{n}$, and $\mathcal{H}^d$ be a $d$-dimensional Hilbert space.
\noindent
\begin{enumerate}
\item If $(C,M_A,M_B)\in\Q$ is a probability distribution obtained by performing binary measurements on a quantum state $\ket{\psi} \in \mathcal{H}^d \otimes \mathcal{H}^d$, then there exists vectors $\vec{a}(x),\vec{b}(y)\in\mathbb{S}_{2d^2}$ such that $C(x,y)=\vec{a}(x)\cdot\vec{b}(y)$.
\item If $\vec{a}(x),\vec{b}(y)$ are unit vectors in $\mathbb{S}_n$, then there exists a probability distribution $(C,0,0)\in\Q$ obtained by performing binary measurements on a maximally entangled state $\ket{\psi} \in \mathcal{H}^{2^{\floor{n/2}}} \otimes \mathcal{H}^{2^{\floor{ n/2}}}$ such that $C(x,y)=\vec{a}(x)\cdot\vec{b}(y)$.
\end{enumerate}
\end{thm}

\begin{cor}
$\Q_0 = \{ C : C(x,y)= \vec{a}(x)\cdot\vec{b}(y),  \norm{\vec{a}(x)}=\norm{\vec{b}(y)} = 1 \,\forall x,y\}$.
\end{cor}

Clearly, $\L\subseteq \Q \subseteq \C$. As first noted by Tsirelson, Grothendieck's inequality~\cite{Grothendieck1953} implies the following statement.
\noindent
\begin{proposition}[\cite{tsi85}]
\label{prop:grothendieck}
$\L_0 \subseteq \Q_0 \subseteq K_G \L_0$, where $K_G$ is Grothendieck's constant.
\end{proposition}

\subsection {Models of communication complexity}
We consider the following model of communication complexity of
non-signaling distributions $\mathbf{p}$.  Alice gets input
$x$, Bob gets input $y$, and after exchanging bits or
qubits, Alice has to output $a$ and Bob $b$ so
that the joint distribution is $p(a,b|x,y)$.
$R_0(\mathbf{p})$
denotes the communication complexity
of simulating $\mathbf{p}$ exactly, using private randomness
and classical communication.
$Q_0(\mathbf{p})$
denotes the communication complexity
of simulating $\mathbf{p}$ exactly, using
quantum communication.
We use superscripts ``$\pub$'' and ``$\ent$'' in the case where
the players share
random bits or quantum entanglement.
For $R_\epsilon(\mathbf{p})$, we are only required to
simulate some distribution $\mathbf{p}'$ such that $\delta(\mathbf{p},\mathbf{p}')\leq\epsilon$, where $\delta(\mathbf{p},\mathbf{p}')=\max\{|p(\mathcal{E}|x,y)-p'(\mathcal{E}|x,y)|:x,y\in\eks\times\wye,\mathcal{E}\subseteq\ay\times\bee\}$ is the total variation distance (or statistical distance) between two distributions.

For distributions with binary outcomes, we write $R_\epsilon(C,M_A,M_B)$ and $Q_\epsilon(C,M_A,M_B)$.
In the case of Boolean functions, $R_\epsilon(C)=R_\epsilon(C,0,0)$
corresponds to the usual notion of computing $f$
with probability at least $1-\epsilon$,
where $C$ is the $\pm 1$ communication matrix of $f$.
From the point of view of communication,
distributions with uniform marginals are the easiest to
simulate.  Suppose we have a protocol that simulates
correlations $C$ with arbitrary marginals.  By using just an additional
shared random bit, both players can flip their outcome
whenever the shared random bit is 1.  Since each players' marginal
outcome is now an even coin flip, this protocol
simulates  the distribution $(C, 0,0)$.
\begin{proposition}
\label{prop:uniform-marginals}
For any Boolean non-signaling distribution $(C, M_A, M_B)$,
we have $R_\epsilon^{\pub}(C, 0, 0)\leq R_\epsilon^{\pub}(C, M_A, M_B)$ and $Q_\epsilon^{\ent}(C, 0, 0)\leq Q_\epsilon^{\ent}(C, M_A, M_B)$.
\end{proposition}

\section{Lower bounds for non-signaling distributions}\label{sec:lower-bounds}

In this section we prove our main theorem, a lower bound on quantum and
classical communication complexity for non-signaling distributions, based
on their affine representations.

Let us define the following quantities, which as we will see may be considered as extensions of the $\cls$ and $\qls$ quantities of~\cite{ls07} (defined in Section~\ref{sec:fact-norms}) to distributions.

\begin{defn}\label{def:lp-quasi-probas}

\begin{itemize}
\item $\cquasi(\mathbf{p}) =
\min \{ \sum_i  \abs{q_i} :
	\exists \mathbf{p}_i\in\L,q_i\in \Real, \mathbf{p} = \sum_i q_i \mathbf{p}_i \} $,
\item $\qquasi(\mathbf{p}) =
\min \{ \sum_i\abs{q_i} :
	\exists \mathbf{p}_i\in \Q, q_i\in \Real, \mathbf{p} = \sum_i q_i \mathbf{p}_i \} $,
\item $\cquasi^\epsilon(\mathbf{p}) = \min \{ \cquasi(\mathbf{p}') :
	\delta(\mathbf{p},\mathbf{p}') \leq \epsilon \} $,
\item $\qquasi^\epsilon(\mathbf{p}) = \min \{ \qquasi(\mathbf{p}') :
	 \delta(\mathbf{p},\mathbf{p}')\leq \epsilon \} $.
\end{itemize}
\end{defn}

Notice that $\sum_i q_i \mathbf{p}_i=\mathbf{p}$ implies in particular $\sum_i q_i=1$.
The quantities $\cquasi(\mathbf{p})$ and $\qquasi(\mathbf{p})$ show how well $\mathbf{p}$ may be represented as an affine combination of local or quantum distributions, a \emph{good} affine combination being one where the sum of absolute values of coefficients $q_i$ is as low as possible.
Figure~\ref{fig:decomposition} represents the decomposition of a
distribution into an affine combination of local distributions.
For a local distribution, we may take positive coefficients $q_i$, and therefore obtain the minimum possible value $\cquasi(\mathbf{p})=1$, and similarly for quantum distributions, so that

\begin{lemma}\label{lem:charact-locality-quanticity}
 $\mathbf{p}\in\L\Longleftrightarrow \cquasi(\mathbf{p})=1$, and
 $\mathbf{p}\in\Q\Longleftrightarrow \qquasi(\mathbf{p})=1$.
\end{lemma}
In other words, the set of local distributions $\L$ form the unit sphere of $\cquasi$, and similarly the set of quantum distributions $\Q$ form the unit sphere of $\qquasi$.
In the binary case, observe that by Proposition~\ref{prop:uniform-marginals}, we have $\qquasi(C) \leq \qquasi(C,u,v)$
and $\cquasi(C) \leq \cquasi(C,u,v)$.
By Proposition~\ref{prop:grothendieck},
$\qquasi(C) \leq \cquasi(C) \leq K_G \qquasi(C)$. Similar properties hold for the approximate versions $\cquasi^\epsilon(C)$ and $\qquasi^\epsilon(C)$.

Our main theorem gives a lower bound on communication complexity in terms
of the quantities $\cquasi$ and $\qquasi$.

\begin{thm}\label{cor:lower-bound-quasi}
For any non-signaling distribution $\mathbf{p}$ and correlation matrix $C$,
\begin{enumerate}
\item $R_0^{\pub}(\mathbf{p}) \geq \log(\cquasi(\mathbf{p}))-1$,
and $R_\epsilon^{\pub}(\mathbf{p})\geq \log(\cquasi^{\epsilon}(\mathbf{p}))-1$.
\item $Q_0^{\ent}(\mathbf{p}) \geq \frac{1}{2}\log(\qquasi(\mathbf{p}))-1$,
and $Q_\epsilon^{\ent}(\mathbf{p})\geq \frac{1}{2}\log(\qquasi^{\epsilon}(\mathbf{p}))-1$.
\item $Q_0^{\ent}(C) \geq \log(\qquasi(C))$,
and $Q_\epsilon^{\ent}(C)\geq \log(\qquasi^{\epsilon}(C))$.
\end{enumerate}
\end{thm}

The proof, minus the details, goes as follows.
Assume that there is a $t$ bit protocol for $\mathbf{p}$.
We derive a noisy, local distribution from $\mathbf{p}$ as follows (Lemma~\ref{lm:decreased-correlations-pi}).
Simulate the protocol, but instead of communicating, guess
a transcript.  If both players agree that this was the correct
transcript, then they output according to $\mathbf{p}$.  This occurs with
probability $2^{-t}$.  Otherwise, output something random.
The resulting distribution is $p'= 2^{-t}\mathbf{p} + (1-2^{-t})\mathbf{q}$ where
$\mathbf{q}$ is some random noise.  But $\mathbf{p}'$ and $\mathbf{q}$ are local, so this
gives an affine representation of $\mathbf{p} = 2^t \mathbf{p}' - 2^t (1-2^{-t})\mathbf{q}$,
showing that $\cquasi(\mathbf{p}) \leq 2^{t+1}-1$.
The rest of this section is devoted to the details.  The only complication
arises from handling arbitrary marginal distributions and
setting up the distribution they should output from when they disagree with
the random transcript.
However, the proof is straightforward, as above, when the marginals are uniform,
which is the case for Boolean functions.

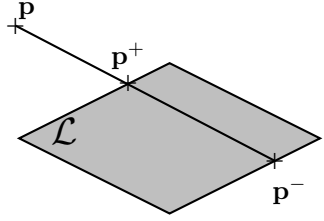
\begin{figure}[h]
\begin{pspicture}(\textwidth,3)
\pspolygon[fillcolor=lightgray, fillstyle=solid](5,1)(7,2)(9,1)(7,0)
\rput (5.6,1.1) {\Large $\mathcal L$}
\psline (4.95,2.5)(8.4,.7)
\rput(6.45, 1.73){\large +}
\rput(8.4,.7){\large +}
\rput(4.95,2.5){\large +}
\rput (5.1,2.7){$\mathbf p$}
\rput(6.45, 2.1){$\mathbf p^+$}
\rput(8.6,.3){$\mathbf p^-$}
\end{pspicture}
\caption{\label{fig:decomposition}$\mathbf p$ is an affine combination of $\mathbf p^+$ and $\mathbf p^-$}
\end{figure}

\subsection{Producing a noisy local distribution {}{from a communication protocol}}

We first show that if a distribution $\mathbf{p}$ may be simulated with $t$ bits of
communication (or $q$ qubits of quantum communication), then there is a noisy version of
this distribution that is local (or quantum).
\begin{lemma}
\label{lm:decreased-correlations-pi}
Let $\mathbf p$ be a non-signaling distribution over $\mathcal A \times \mathcal B$
with input set $\mathcal X \times \mathcal Y$.
\begin{enumerate}
\item Assume that $R_0^{\pub}(\mathbf p) \leq t$, then
there exist two marginal distributions $p_A(a|x)$ and $p_B(b|y)$ such that the distribution
$
p_l(a,b|x,y)=\frac{1}{2^t}p(a,b|x,y)+(1-\frac{1}{2^t})p_A(a|x)p_B(b|y)
$
is local.
\item Assume that $Q_0^{\ent}(\mathbf p)\leq q$, then
there exist two marginal distributions $p_A(a|x)$ and $p_B(b|y)$ such that the distribution
$
p_l(a,b|x,y)=\frac{1}{2^{2q}}p(a,b|x,y)+(1-\frac{1}{2^{2q}})p_A(a|x)p_B(b|y)
$
is quantum.
\item Assume that $\mathbf p=(C,0,0)$ and $Q_0^{\ent}(C) \leq q$, then $C / {2^q} \in \Q_0$.
\end{enumerate}
\end{lemma}
\begin{proof}
We assume that the length of the transcript is exactly t bits for
each execution of the protocol, adding dummy bits if necessary.
We now fix some notations. In the original protocol, the players pick a random string
$\lambda$ and exchange some communication
whose transcript is denoted $T(x,y,\lambda)$.
Alice then outputs some value $a$ according to a probability distribution
$p_P(a|x, \lambda, T)$. Similarly,
Bob outputs some value $b$ according to a probability distribution
$p_P(b|y, \lambda, T)$.

From Alice's point of view, on input $x$ and shared randomness $\lambda$,
only a subset of the set of all $t$-bit transcripts can be produced:
the transcripts $S\in\{0,1\}^t$ for which there exists a $y$ such that $S=T(x,y,\lambda)$.
We will call these transcripts the set of valid transcripts for $(x,\lambda)$.
The set of valid transcripts for Bob is defined similarly. We denote these sets
respectively $U_{x,\lambda}$ and $V_{y,\lambda}$.

We now define a local protocol for the distribution $p_l(a,b|x,y)$:
\begin{itemize}
\item As in the original protocol, Alice and Bob initially share some random string $\lambda$.
\item Using additional shared randomness, Alice and Bob choose a transcript $T$ uniformly at random in $\{0,1\}^t$.
\item If $T$ is a valid transcript for $(x,\lambda)$, she outputs $a$
according to the distribution $p_P(a|x,\lambda,T)$.
If it is not, Alice outputs $a$ according to a distribution $p_A(a|x)$
which we will define later.
\item Bob does the same. We will also define the distribution $p_B(b|y)$ later.
\end{itemize}

Let $\mu$ be the distribution over the randomness and the $t$-bit strings in the
local protocol.
By definition, the distribution produced by this protocol is
\begin{eqnarray*}
 p_l(a,b|x,y)&=&
 \sum_\lambda \mu(\lambda) \left[
 \sum_{T\in U_{x,\lambda} \cap V_{y,\lambda}} \mu(T) p_P(a|x,\lambda, T) p_P(b|y,\lambda, T) +
p_B(b|y) \sum_{T\in U_{x,\lambda} \cap \bar V_{y,\lambda}} \mu(T) p_P(a|x,\lambda, T) \right.\\
&+&\left. p_A(a|x) \sum_{T\in \bar U_{x,\lambda} \cap V_{y,\lambda}} \mu(T) p_P(b|y,\lambda, T) +
p_B(b|y) p_A(a|x) \sum_{T\in \bar U_{x,\lambda} \cap \bar V_{y,\lambda}} \mu(T) \right]
\end{eqnarray*}

We now analyze each term separately. For fixed inputs $x,y$ and
shared randomness $\lambda$, there is only one transcript which is valid for both Alice and Bob, and when they use this transcript for each $\lambda$, they output according to the distribution $\mathbf p$.
Therefore, we have
$$ \sum_\lambda \mu(\lambda)
 \sum_{T\in U_{x,\lambda} \cap V_{y,\lambda}} \mu(T) p_P(a|x,\lambda, T) p_P(b|y,\lambda, T) = \frac 1 {2^t} p(a,b|x,y).$$

Let $A_x$ be the event that Alice's transcript is valid for $x$ (over random $\lambda,T$), and $\bar A_x$ its negation (similarly $B_y$ and $\bar B_y$ for Bob).
We denote $$p_P(a|x, A_x \cap \bar B_y)= \frac {\sum_\lambda \mu(\lambda)
\sum_{T \in U_{x,\lambda} \cap \bar V_{y,\lambda}} \mu(T) p_P(a|x, \lambda, T)}
{\mu(A_x \cap \bar B_y)},$$
where, by definition, we have
$\mu(A_x\cap \bar B_y)=\sum_\lambda \mu(\lambda) \sum_{T\in U_{x,\lambda} \cap \bar V_{y,\lambda}} \mu(T)$.
We will show that this distribution is independent of $y$ and that the corresponding distribution $p_P(b|y, \bar A_x \cap B_y)$ for Bob is independent of $x$.
Using these distributions, we may write $p_l(a,b|x,y)$ as
\begin{eqnarray*}
 p_l(a,b|x,y)&=&
 \frac 1 {2^t} p(a,b|x,y) + \mu(A_x \cap \bar B_y) p_B(b|y)  p_P(a|x,A_x\cap \bar B_y)\\
&+& \mu(\bar A_x \cap B_y) p_A(a|x) p_P(b|x,\bar A_x\cap B_y) + \mu(\bar A_x \cap \bar B_y) p_B(b|y) p_A(a|x)
\end{eqnarray*}
Summing over $b$, and using the fact that $\mathbf p_l$ and $\mathbf p$
are non-signaling, we have
\begin{eqnarray*}
 p_l(a|x)&=&\frac{1}{2^t}p(a|x)+\mu(A_x\cap \bar{B}_y) p_P(a|x, A_x\cap \bar{B}_y)\\
&+&\mu(\bar{A}_x\cap {B}_y) p_A(a|x)+\mu(\bar{A}_x\cap \bar{B}_y)p_A(a|x)\\
&=&\frac{1}{2^t}p(a|x)+\mu(A_x \cap \bar{B}_y) p_P(a|x, A_x \cap \bar{B}_y)+\mu(\bar{A}_x)p_A(a|x),
\end{eqnarray*}
Note that by definition, $\mu(A_x)=\sum_\lambda \mu(\lambda) \sum_{T\in U_{x,\lambda}} \mu(T)$ is independent of $y$, therefore so is $\mu(A_x \cap \bar B_y) = \mu(A_x) - \mu(A_x \cap B_y)= \mu(A_x) -\frac 1 {2^t}$.
From the expression for $p_l(a|x)$, we can conclude that $p_P(a|x,A_x \cap \bar B_y)$
is independent of $y$ and can be evaluated by Alice
(and similarly for the analogue distribution for Bob). We now set
\begin{eqnarray*}
p_A(a|x)&=&p_P(a|x,A_x\cap \bar{B}_y)\\
p_B(b|y)&=&p_P(b|y, \bar{A}_x\cap {B}_y).
\end{eqnarray*}
Therefore, the final distribution obtained from the local protocol may be written as
\begin{eqnarray*}
 p_l(a,b|x,y)&=&\frac{1}{2^t}p(a,b|x,y)+\mu(A_x\cap \bar{B}_y)p_A(a|x)p_B(b|y)\\
&+& \mu(\bar{A}_x\cap {B}_y)p_A(a|x)p_B(b|y)+\mu(\bar{A}_x\cap\bar{B}_y)p_A(a|x)p_B(b|y)\\
&=&\frac{1}{2^t}p(ab|xy)+(1-\frac{1}{2^t})p_A(a|x)p_B(b|y).
\end{eqnarray*}

For quantum protocols, we first simulate quantum communication using shared entanglement and teleportation, which uses 2 bits of classical communication for each qubit. Starting with this protocol using $2q$ bits of classical communication, we may use the same idea as in the classical case, that is choosing a random $2q$-bit string interpreted as the transcript, and replacing the players' respective outputs by independent random outputs chosen according to $p_A$ and $p_B$ if the random transcript does not match the bits they would have sent in the original protocol.

In the case of binary outputs with uniform marginals, that is, $\mathbf{p}=(C,0,0)$, we may improve the exponent of the scaling-down coefficient $2^{2q}$ by a factor of $2$ using a more involved analysis and a variation of a result by~\cite{kremer, yao, ls07} (the proof is given in Appendix~\ref{appendix:quantum-communication} for completeness).
\begin{lemma}[\cite{kremer, yao, ls07}]
\label{lemma:quantum-communication}
Let $(C, M_A, M_B)$ be a distribution simulated by a
quantum protocol with shared entanglement
using $q_A$ qubits of communication from Alice to Bob and $q_B$ qubits from Bob to Alice.
There exist vectors $\vec{a}(x),\vec{b}(y)$ with
$\norm{\vec{a}(x)}\leq 2^{q_B}$ and $\norm{\vec{b}(y)} \leq 2^{q_A}$ such that
$C(x,y) = \vec{a}(x)\cdot\vec{b}(y)$.
\end{lemma}
The fact that $C/2^q \in\Q_0$ then follows from Theorem~\ref{thm:tsirelson} part~2.
\end{proof}

\subsection{Deriving {}{an affine model and} the lower bound from the {noisy distribution}}

In this section we show that using Lemma~\ref{lm:decreased-correlations-pi}, an explicit affine model can be derived from a (classical or quantum) communication protocol for $\mathbf{p}$,
which gives us a lower bound technique for communication
complexity in terms of how ``good'' the affine model is.
We now are ready to complete the proof of Theorem~\ref{cor:lower-bound-quasi}.

\begin{proof}[Proof of Theorem~\ref{cor:lower-bound-quasi}]
We give a proof for the classical case, the quantum case follows the same lines.
Let $c$ be the number of bits exchanged. From Lemma~\ref{lm:decreased-correlations-pi}, we know that
there exists marginal distributions $p_A(a|x)$ and $p_B(b|y)$ such that
$p_l(a,b|x,y)=\frac{1}{2^t}p(a,b|x,y)+(1-\frac{1}{2^t})p_A(a|x)p_B(b|y)$ is local.
This gives an affine model for $p(a,b|x,y)$, as the following combination of two local distributions:
$$p(a,b|x,y)=2^t p_l(a,b|x,y) + (1-2^t) p_A(a|x) p_B(b|y).$$
Then $\cquasi(\mathbf{p}) \leq 2^{t+1} -1$.

In the case of binary outputs with uniform marginals, $\mathbf{p}_l=( C/ 2^t,0,0)$, and Lemma~\ref{lm:decreased-correlations-pi} implies that $ C/2^t \in \L_0$.
By following the local protocol for $ C/2^t$ and letting Alice flip her output, we also get a local protocol for $- C/2^t$, so $-C/2^t \in \L_0$ as well. Notice that we may build an affine model for $C$ as a combination of $C/2^t$ and $-C/2^t$:
 $$C
= \frac{1}{2}(2^t+1)\frac{C}{2^t} - \frac{1}{2}(2^t-1)\frac{C}{2^t}.$$
Then, $\cquasi(C) \leq {2^t}$.
\end{proof}
\subsection{Factorization norm and related measures}\label{sec:fact-norms}

In the special case of distributions over binary variables with
uniform marginals, the quantities~$\cquasi$ and~$\qquasi$
become equivalent to the original quantities defined in~\cite{lmss07,ls07}
(at least for the interesting case of non-local correlations,
that is correlations with non-zero communication complexity).
When the marginals are uniform we omit them and
write $\cquasi(C)$ and $\qquasi(C)$.
The following are reformulations as Minkowski functionals
of the definitions appearing in~\cite{lmss07,ls07}.

\begin{defn}
\begin{itemize}
 \item $\cls(C) = \min \{ \Lambda>0 :\frac{1}{\Lambda}C \in \L_0\} $,
 \item $\qls(C) = \min \{ \Lambda>0 :\frac{1}{\Lambda}C \in \Q_0\} $,
 \item $\cls^\alpha(C) = \min \{ \cls(C'): 1\leq C(x,y)C'(x,y)\leq \alpha,\ \forall x,y\in\eks\times\wye\} $,
 \item $\qls^\alpha(C) = \min \{ \qls(C'): 1\leq C(x,y)C'(x,y)\leq \alpha,\ \forall x,y\in\eks\times\wye\} $.
\end{itemize}
\end{defn}

\begin{thm}
\label{lemma:cmargvscquasi}
For any correlation matrix $C:\eks\times \wye \rightarrow [-1,1]$,
\begin{enumerate}
\item $\cquasi(C) =1$ iff $\cls(C)\leq 1$, and
$\qquasi(C) =1$ iff $\qls(C)\leq 1$,
\item $\cquasi(C)>1\Longrightarrow \cmarg(C)=\cquasi(C)$,
\item $\qquasi(C)>1\Longrightarrow \qmarg(C)=\qquasi(C)$.
\end{enumerate}
\end{thm}

\begin{proof}
The first item follows by definition of $\cls$ and $\qls$.
For the next items, we give the proof for $\cls$, and the proof for
$\qls$ is similar. The key to the proof is that
if $C\in\L_0$, then $-C\in\L_0$ (it suffices for one of the players to flip his output).

[$\cquasi(C)\leq \cmarg(C)$]
If $\cquasi(C)>1$, then $\Lambda=\cmarg(C) >  1$.  Let
$C^+=\frac{C}{\Lambda}$ and $C^-=-\frac{C}{\Lambda}$. By definition of $\cmarg(C)$, both $C^+$ and $C^-$ are in $\L_0$. Furthermore, let $q_+=\frac{1+\Lambda}{2}\geq 0$ and $q_-=\frac{1-\Lambda}{2}\leq  0$. Since $C=q_+C^++q_-C^-$, this determines an affine model for $C$ with $|q_+|+|q_-|=\Lambda$.

[$\cquasi(C)\geq \cmarg(C)$] Let $\Lambda=\cquasi(C)$. By definition of $\cquasi(C)$, there exists $C_i$ and $q_i$ such that $C=\sum_i q_i C_i$ and $\Lambda=\sum_i|q_i|$. Let $\tilde{C}_i=\sgn(q_i)C_i$ and $p_i=\frac{|q_i|}{\Lambda}$. Then, $\frac{C}{\Lambda}=\sum_i p_i \tilde{C}_i$ and therefore $\frac{1}{\Lambda}C\in \L_0$ since $\tilde{C}_i\in \L_0$.
\end{proof}

In the special case of sign matrices (corresponding to Boolean functions, as shown above), we also have the following correspondence between $\cquasi^\epsilon,\qquasi^\epsilon$, and $\cls^\alpha,\qls^\alpha$.
\begin{thm}
 \label{lem:epsilon-alpha}
Let $0\leq\epsilon< 1/2$ and $\alpha=\frac{1}{1-2\epsilon}$.
For any sign matrix $C:\eks\times \wye \rightarrow \{-1,1\}$,
\begin{enumerate}
\item $\cquasi^\epsilon(C)>1\Longrightarrow \cls^\alpha(C)=\frac{\cquasi^\epsilon(C)}{1-2\epsilon}$,
\item $\qquasi^\epsilon(C)>1\Longrightarrow \qls^\alpha(C)=\frac{\qquasi^\epsilon(C)}{1-2\epsilon}$.
\end{enumerate}
\end{thm}
\begin{proof}
 We give the proof for $\cls^\alpha$, the proof for $\qls^\alpha$ is similar.

[$\cls^\alpha(C)\leq\frac{\cquasi^\epsilon(C)}{1-2\epsilon}$]
By definition of $\cquasi^\epsilon(C)$, there exists a correlation matrix $C'$ such that $\cquasi(C')=\cquasi^\epsilon(C)$ and $|C(x,y)-C'(x,y)|\leq 2\epsilon$ for all $x,y\in\eks\times\wye$. Since $C$ is a sign matrix, and $C'$ is a correlation matrix,
$\sgn(C'(x,y))=C(x,y)$ and $1-2\epsilon \leq |C'(x,y)|\leq 1$.
Hence
$1\leq C(x,y)\frac{C'(x,y)}{1-2\epsilon}\leq\frac{1}{1-2\epsilon}= \alpha$.
This implies that $\cls^\alpha(C)\leq\cls(\frac{C'}{1-2\epsilon})=
\frac{\cls(C')}{1-2\epsilon}=\frac{\cquasi(C')}{1-2\epsilon}$,
where we used the fact that $\cls(C')=\cquasi(C')$ since $\cquasi(C')>1$.

[$\cls^\alpha(C)\geq\frac{\cquasi^\epsilon(C)}{1-2\epsilon}$]
By definition of $\cls^\alpha(C)$, there exists a
(not necessarily correlation) matrix $C'$ such that
$\cls(C')=\cls^\alpha(C)$ and $1\leq C(x,y)C'(x,y)\leq \alpha$
for all $x,y$.
Since $C$ is a sign matrix, this implies $\sgn(C'(x,y))=C(x,y)$
and $1-2\epsilon \leq |\frac{C'(x,y)}{\alpha}|\leq  1$.
Therefore, $|C(x,y)-\frac{C'(x,y)}{\alpha}|\leq 2\epsilon$ for all $x,y$.
This implies that $\cquasi^\epsilon(C)\leq\cquasi(\frac{C'}{\alpha})
=\cls(\frac{C'}{\alpha})=(1-2\epsilon)\cls(C')$,
where we have used the fact that $\cquasi(\frac{C'}{\alpha})
=\cls(\frac{C'}{\alpha})$ since $\cquasi(\frac{C'}{\alpha})\geq\cquasi^\epsilon(C)>1$.
\end{proof}

\noindent {\bf Discussion.}
Just as the special case $\cls(C)$, $\cquasi(\mathbf{p})$ may be expressed as a linear program. However, while $\qls(C)$ could be expressed as a semidefinite program, this may not be true in general for $\qquasi(\mathbf{p})$ (even though it can still be studied by SDP relaxation, as shown in~\cite{npa08,DLTW}).

Lemmas~\ref{lemma:cmargvscquasi} and~\ref{lem:epsilon-alpha} establish that
Corollary~\ref{cor:lower-bound-quasi}
is a generalization of Linial and Shraibman's factorization norm lower bound technique.
Note that Linial and Shraibman use $\qls^\alpha$ to derive a lower bound not only on the quantum communication complexity $Q_\epsilon^{\ent}$, but also on the classical complexity $R_\epsilon^{\pub}$. In the case of binary outcomes with uniform marginals (which includes Boolean functions, studied by Linial and Shraibman, as a special case), we obtain a similar result by combining our bound for $Q_\epsilon^{\ent}(C)$ with the fact that $Q_\epsilon^{\ent}(C)\leq \lceil{\smfrac{1}{2}R_\epsilon^{\pub}(C)}\rceil$, which follows from superdense coding. This implies $R_\epsilon^{\pub}(C) \geq 2\log(\qls^\epsilon(C))-1$. In the general case, however, we can only prove that $R_\epsilon^{\pub}(\mathbf{p}) \geq \log(\qls^\epsilon(\mathbf{p}))-1$.
This may be due to the fact that the result holds in the much more general
setting of non-signaling distributions with arbitrary outcomes and marginals.

Because of Proposition~\ref{prop:grothendieck}, we know that
$\cls(C) \leq K_G \qls(C)$ for correlations.
Note also that although $\qls$ and $\cls$ are matrix norms, this fails
to be the case for $\qquasi$ and $\cquasi$, even in the case of correlations.  Nevertheless, it is still possible to formulate dual quantities,
which turn out to have sufficient structure, as we show in the next section.

\section{Duality, Bell inequalities, and XOR games}\label{sec:dual}

In their primal formulation, the $\qquasi$ and $\cquasi$ methods are
difficult to apply since they are formulated as a minimization problem.
Transposing to the dual space not only turns the method into
a maximization problem; we show it also has a very natural, well-understood
interpretation since it coincides with maximal violations
of Bell and Tsirelson inequalities.
This is particularly relevant to physics, since it formalizes in very precise terms
the intuition that distributions with large Bell inequality
violations should require more communication to simulate.

Recall that for any norm $\norm{\cdot}$ on a vector space $V$, the dual norm  is
$\norm{B}^*= \max_{v\in V:\norm{v}\leq 1} B(v)$, where $B$ is a linear functional on $V$.

\subsection{Bell and Tsirelson inequalities}

Bell inequalities were first introduced by Bell~\cite{bell64}, as bounds on the correlations that could be achieved by any \emph{local} physical theory. He showed that
quantum correlations could violate these inequalities and therefore exhibited
non-locality. Tsirelson later proved that quantum correlations should also respect
some bound (known as the Tsirelson bound), giving a first example of a
``Tsirelson-like'' inequality for quantum distributions~\cite{tsirelson80}.

Since the set of non-signaling distributions $\C$ lies in an affine space $\aff(\C)$, we may consider the isomorphic dual space of linear functionals over this space.
The dual quantity $\cquasi^*$ (technically not a dual norm since $\cquasi$ itself is not a norm in the general case) is the maximum value of a linear
functional in the dual space on local distributions, and $\qquasi^*$ is
the maximum value of a linear functional on quantum distributions.
These are exactly what is captured by the Bell and Tsirelson inequalities.
\begin{defn}[Bell and Tsirelson inequalities]
\label{defn:bell}
Let $B:\aff(\C)\mapsto\Real$ be a linear functional on the (affine hull of the) set of non-signaling distributions,
$B(\mathbf{p})=\sum_{a,b,x,y} B_{abxy} p(a,b|x,y)$.
Define $\cquasi^*(B)=\max_{\mathbf{p}\in\L}|B(\mathbf{p})|$ and $\qquasi^*(B)=\max_{\mathbf{p}\in\Q}|B(\mathbf{p})|$. 
A Bell inequality is a linear inequality satisfied by any local distribution:
$$B(\mathbf{p})\leq \cquasi^*(B)\ (\forall\ \mathbf{p}\in\L),$$
and a Tsirelson inequality is a linear inequality satisfied by any quantum distribution: $$B(\mathbf{p})\leq \qquasi^*(B)\ (\forall\ \mathbf{p}\in\Q).$$
\end{defn}

By linearity (Proposition~\ref{prop:representation})
Bell inequalities are often expressed as linear functionals
over the correlations in the case of binary outputs and uniform marginals.

Finally, $\qquasi$ and $\cquasi$
amount to finding a maximum violation of a
(normalized) Bell or Tsirelson inequality.
\begin{thm}\label{thm:lp-bell} For any distribution $\mathbf{p}\in \mathcal{C}$,
\begin{enumerate}
 \item \label{item:1}$\cquasi(\mathbf{p})=\max \{ B(\mathbf{p}): \forall \mathbf{p}'\in \L,\  \abs{B(\mathbf{p}')}\leq 1 \}$, and
\item \label{item:2}$\qquasi(\mathbf{p})=\max \{ B(\mathbf{p}): \forall \mathbf{p}'\in \Q,\ \abs{B(\mathbf{p}')} \leq 1 \}$,
\end{enumerate}
where the maximization is over linear functionals $B:\aff(\C)\mapsto\Real$.
\end{thm}

\begin{proof}
The proof of item~\ref{item:1} follows by LP duality from the definition of $\cquasi$.
Nevertheless, we give an alternative proof that
can be easily adapted to prove item~\ref{item:2} (it suffices to replace $\cquasi$ by $\qquasi$ and $\L$ by $\Q$).
The key idea of the proof is to use the convex conjugate of $\cquasi$ (written $\cquasi^\Star$)  which is
closely related to the dual expression (written $\cquasi^*$) , and apply it twice.

We first recall basic facts about convex conjugate functions (See~\cite{Boyd2004} for full details).
For a function $f: \Real^n \rightarrow \Real$, the convex conjugate
function  $f^\Star: \Real^n \rightarrow \Real$ is defined as:
$$f^\Star(y)=\sup_{x \in {\rm dom}(f)} (y^T x - f(x)),$$
where ${\rm dom}(f)$ denotes the domain of $f$.
It is known that $f^{\Star\Star}=f$ provided that $f$ is convex and closed {\it i.e.},~its epigraph is closed.

By grouping negative and positive terms together, it is easy to see that
$\cquasi(\mathbf p)=\min\{ k^+ + k^- : k^+,k^-\in\Real^+,\exists \mathbf p^+, \mathbf p^- \in \L,\mathbf p=k^+ \mathbf p^+ - k^- \mathbf p^-\}$. We consider $\cquasi$ as a function over $\aff(\L)$. Then, it is straightforward to verify that $\cquasi$ is convex, and since its domain $\aff(\L)$ is closed, $\cquasi$ is also a closed function.

We then have by definition
\begin{eqnarray*}
\cquasi^\Star(B) &=& \max_{\mathbf p \in \aff(\L)} (B(\mathbf p) - \cquasi(\mathbf p)),\\
&=& \max_{\mathbf p_1, \mathbf p_2 \in \L, k_1 - k_2 =1}
( B(k_1 \mathbf p_1 - k_2 \mathbf p_2) - (k_1 + k_2)),\\
&=&  \max_{\mathbf p_1, \mathbf p_2 \in \L, k_1 - k_2 =1}
( k_1 (B(\mathbf p_1)-1) - k_2 (B(\mathbf p_2) + 1) ).
\end{eqnarray*}
Therefore,
\begin{eqnarray*}
\cquasi^\Star(B)&=
\begin{cases}
 \max_{\mathbf p \in \L} \abs{B(\mathbf p)}-1 & \text { if } \max_{\mathbf p \in \L} \abs{B(\mathbf p)} \leq 1,\\
+\infty & \text{ otherwise}.
\end{cases}
\end{eqnarray*}
Taking the convex conjugate a second time, we obtain
\begin{align*}
 \cquasi^{\Star\Star}(\mathbf p)=\max_B\ (B(\mathbf p)-\cquasi^\Star(B)).
\end{align*}
From the expression for $\cquasi^\Star(B)$ above, it is clear that the maximum is achieved for a linear functional $B$ such that $\max_{\mathbf p \in \L} \abs{B(\mathbf p)} \leq 1$. Let the maximum be achieved by a linear functional $\bar{B}$, and let us consider $\bar{B}_{\max}=\max_{\mathbf p \in \L} \bar{B}(\mathbf p)$ and $\bar{B}_{\min}=\min_{\mathbf p \in \L} \bar{B}(\mathbf p)$. We show that we can assume without loss of generality that $\abs{\bar{B}_{\min}}\leq\bar{B}_{\max} = 1$. Indeed, we must have $\abs{\bar{B}_{\min}}\leq \bar{B}_{\max}$, otherwise $\bar{B}$ could not achieve the maximum since $-\bar{B}$ would yield a larger value. This implies that $\cquasi^\Star(\bar{B})=\bar{B}_{\max}-1$ and $\cquasi^{\Star\Star}(\mathbf p)=\bar{B}(\mathbf p)-\bar{B}_{\max}+1$. Then, the maximum is also achieved by the linear functional $\bar{B}'(\mathbf p)=\bar{B}(\mathbf p)-\bar{B}_{\max}+1$, which satisfies $\bar{B}'_{\max}=1$ and therefore $\cquasi^\Star(\bar{B}')=0$. From the expression for $\cquasi^{\Star\Star}$, we therefore obtain
$\cquasi^{\Star\Star}(\mathbf p) = \cquasi(\mathbf{p})=\max \{ B(\mathbf{p}): \forall \mathbf{p}'\in \L,\  |B(\mathbf{p}')|\leq 1 \}$.

\end{proof}

\subsection{XOR games and Bell inequalities for correlations}
In the special case of XOR games, there is a close connection between
winning probability and Bell inequalities, which we make explicit in this section.

In an XOR game, Alice is given some input $x$ and Bob is given an input $y$,
and they should output $a=\pm 1$ and $b=\pm 1$. They win if  $a\cdot b$ equals some
$\pm1$ function $G(x,y)$. Since they are not allowed to communicate, their strategy may be represented as a local correlation matrix $S\in\L_0$. We consider the distributional
version of this game, where $\mu$ is a distribution on the inputs.
The winning bias given some strategy $S$ with respect to $\mu$
is
$\epsilon_\mu(G{\parallel}S) = \sum_{x,y} \mu(x,y) G(x,y) S(x,y)$,
and $\epsilon_\mu^{\pub}(G) = \max_{S\in \L_0} \epsilon_\mu(G{\parallel}S)$
is the maximum winning bias of any local (classical) strategy
(for convenience, we consider the bias instead of game value $\omega_\mu^{\pub}(G)=(1+\epsilon_\mu^{\pub}(G))/2$).  We define  $\epsilon_\mu^\ent(G)$ similarly for quantum strategies.
When the input distribution is not fixed, we define the game biases
as $\epsilon^{\pub}(G)=\min_\mu\epsilon_\mu^{\pub}(G)$
and $\epsilon^\ent(G)=\min_\mu\epsilon^\ent_\mu(G)$.
\begin{lemma}
There is a bijection between XOR games $(G,\mu)$ and normalized correlation Bell inequalities.
\end{lemma}
\begin{proof}
For a given XOR game $G$, and a local strategy $C$, its winning probability, or more
simply its bias, can be written as a linear equation, which we write
$G{{\circ}}\mu\,(C)= \epsilon_\mu(G{\parallel}C) $
where ${\circ}$ is the Hadamard (entrywise) product.
This can be seen as a linear functional over the space of strategies.
By Definition~\ref{defn:bell}, $\cls^*(G{\circ} \mu)=\epsilon_\mu^{\pub}(G)$, and
$\epsilon_\mu(G{\parallel}C) \leq\epsilon_\mu^{\pub}(G)$ is a Bell inequality satisfied by any local correlation matrix $C$.
Similarly, when the players are allowed to use entanglement, we get
 a Tsirelson inequality on quantum correlations,
$ \epsilon_\mu(G{\parallel}C) \leq \epsilon^\ent_\mu(G)$
(the quantum bias is also equivalent to a dual
norm $\epsilon^\ent_\mu(G)=\qls^*(G{\circ}\mu)$).

Conversely, consider a general linear functional
$B(C)=\sum_{x,y}B_{xy}C(x,y)$ on $\aff(\C_0)$, defining a correlation Bell inequality
$B(C)\leq \cls^*(B)\ \forall\ C\in \L_0$. Dividing this Bell inequality by $N=\sum_{x,y}|B_{xy}|$, we see that it determines
an XOR game specified by a sign matrix $G(x,y)=\sgn(B_{xy})$ and an input distribution
$\mu_{xy}=\frac{|B_{xy}|}{N}$, and having a game bias $\epsilon_\mu^{\pub}(G)=\frac{\cls^*(B)}{N}$.
\end{proof}
By Theorem~\ref{thm:lp-bell} and the previous bijection
(see also Lee \textit{et al.}~\cite{lss08}):
\begin{cor}\label{cor:LS-game}
\begin{enumerate}
\item $\cls(C)=\max_{\mu,G}\frac{\epsilon_\mu(G{\parallel}C)}{\epsilon_\mu^{\pub}(G)}$,
\item $\cls(C)\geq\frac{1}{\epsilon^{\pub}(C)}$.
\end{enumerate}
\end{cor}

The second part follows by letting $G=C$.
Even though playing correlations $C$ for a game $G=C$ allows us to win with probability one, there are cases where some other game $G\neq C$ yields a larger ratio. In these cases, we have $\cls(C)>\frac{1}{\epsilon^{\pub}(C)}$ so that $\cls$ gives a stronger lower bound for communication complexity than the game value (which has been shown to be equivalent to the discrepancy method~\cite{lss08}).
Similar properties hold for the quantum values, in particular, we have $\qls(C)\geq\frac{1}{\epsilon^{\ent}(C)}$.

We can characterize when the inequality is tight.
Let $\epsilon^{\pub}_=(C)=
	\max_{S\in\L_0} \{ \beta : \forall x,y,  C(x,y)S(x,y){=}\beta\} $,
that is, we only consider strategies that win the game with equal
bias with respect to all distributions.
For the sake of comparison, the game bias may also be expressed as~\cite{vonneumann28}:
$$\epsilon^{\pub}(C) =\max_{S\in \L_0} \{ \beta : \forall x,y,  C(x,y)S(x,y){\geq}\beta\} = \max_{S\in \L_0} \min_{x,y}  C(x,y)S(x,y).$$
\begin{lemma}
 ${\cls(C)} = \frac{1}{\epsilon^{\pub}_=(C)} $.

\end{lemma}

We can also relate the game value to $\cls^\alpha(C)$, as it was shown in~\cite{lss08} that for $\alpha\to\infty$, $\cls^\infty(C)$ is exactly the inverse of the game bias $\frac{1}{\epsilon^{\pub}(C)}$. We show that this holds as soon as $\alpha=\frac{1}{1-2\epsilon}$ is large enough for $C$ to be local up to an error $\epsilon$, completing the picture given in Lemma~\ref{lem:epsilon-alpha}.
\begin{lemma}\label{lem:gamma2-infinity}
Let $0\leq\epsilon< 1/2$ and $\alpha=\frac{1}{1-2\epsilon}$.
For any sign matrix $C:\eks\times \wye \rightarrow \{-1,1\}$,
\begin{enumerate}
 \item $\cquasi^\epsilon(C)=1
\Longleftrightarrow \epsilon \geq 1 - \omega^{\pub}(C)
\Longleftrightarrow\alpha\geq\frac{1}{\epsilon^{\pub}(C)}
\Longleftrightarrow\cls^\alpha(C)=\cls^\infty(C)=\frac{1}{\epsilon^{\pub}(C)}$,
\item $\qquasi^\epsilon(C)=1
\Longleftrightarrow \epsilon \geq 1 - \omega^\ent(C)
\Longleftrightarrow\alpha\geq\frac{1}{\epsilon^\ent(C)}
\Longleftrightarrow\qls^\alpha(C)=\qls^\infty(C)=\frac{1}{\epsilon^\ent(C)}$.
\end{enumerate}
\end{lemma}

\begin{proof}[Proof]
By von Neumann's minmax principle~\cite{vonneumann28},
\begin{eqnarray*}
\epsilon^{\pub}(C)& =&  \max_{S\in \L_0} \min_{x,y}  C(x,y)S(x,y) \\
& =&  \max_{S\in \L_0} \min_{x,y}  1- |C(x,y) - S(x,y)| \\
\end{eqnarray*}
where we used the fact that $C$ is a sign matrix. This implies that $\cquasi^\epsilon(C)=1\Leftrightarrow\epsilon\geq \frac{1-\epsilon^{\pub}(C)}{2}\Leftrightarrow\alpha\geq\frac{1}{\epsilon^{\pub}(C)}$.

By Lemma~\ref{lem:epsilon-alpha}, this in turn implies that $\cls^\alpha(C)=\frac{\cquasi^\epsilon(C)}{1-2\epsilon}$ for all $\epsilon<\frac{1-\epsilon^{\pub}(C)}{2}$. By continuity, taking the limit $\epsilon\to\frac{1-\epsilon^{\pub}(C)}{2}$ yields $\cls^\alpha(C)=\frac{1}{\epsilon^{\pub}(C)}$ for $\alpha=\frac{1}{\epsilon^{\pub}(C)}$. From~\cite{lss08}, $\cls^\infty(C)=\frac{1}{\epsilon^{\pub}(C)}$, and the lemma follows by the monotonicity of $\cls^\alpha(C)$ as a function of $\alpha$.
\end{proof}

\section{Bounding the violation of Bell inequalities}
\label{sec:gamma-vs-nu}
In this section, we give bounds on the maximal violations of Bell inequalities. By Theorem~\ref{thm:lp-bell}, this is equivalent to bounding the ratio between $\qquasi$ and  $\cquasi$.
In the case of distributions over binary outcomes with uniform marginals (correlations),
the theorems of Tsirelson (Theorem~\ref{thm:tsirelson})
and Grothendieck (Proposition~\ref{prop:grothendieck}) imply that
$\qls$ and $\cls$ differ by at most a constant.
This is bad news for anyone trying to find a Boolean function
with high randomized communication complexity and considerably smaller
quantum communication complexity, since it means that any randomized lower
bound obtained by using $\nu$ will yield a similar quantum lower bound.
Although neither of these theorems are known to hold beyond the Boolean setting with uniform marginals, we show in this section that this surprisingly also extends to
non-signaling distributions.
This is also bad news for anyone looking for
large Bell inequality violations by quantum distributions, since in this
case, $\qquasi(\mathbf{p})=1$, and the maximum Bell inequality we can hope
for will be bounded above by the expressions below.

\begin{thm}\label{thm:nu-gamma2}
For any distribution $\mathbf{p}\in\C$, with inputs in
$\eks \times  \wye$ and outcomes in $\A \times \B$ with $A=|\A|, B=|\B|$,
\begin{enumerate}
 \item $\cquasi(\mathbf{p})\leq (2K_G+1)\qquasi(\mathbf{p})$ when $A=B=2$,
 \item $\cquasi(\mathbf{p})\leq[2AB(K_G+1)-1]\qquasi(\mathbf{p})$ for any $A,B$.
\end{enumerate}
\end{thm}

Therefore, one cannot hope to prove
separations between classical and quantum communication
using this method, except in the case where the number of
outcomes is large.  For binary outcomes at least, this says
that arguments based on analyzing the distance to the quantum set
only, without taking into account the particular structure of
the distribution, will not suffice to prove large separations;
and other techniques, such as information theoretic arguments,
may be necessary.

For example, Brassard \textit{et al.}~\cite{bct99} give a (promise) distribution
based on the Deutsch-Jozsa problem,
which can be obtained exactly with entanglement and no communication,
but which requires linear communication to simulate exactly.
The lower bound is proven using a corruption bound~\cite{bcw98},
which is closely related to the information theoretic subdistribution
bound~\cite{jkn08}.
For this problem, $\eks=\wye=\{0,1\}^n$ and $\A=\B=[n]$,
therefore our method can only prove a lower bound logarithmic in $n$.
This is the first example of a problem for which the
corruption bound gives an exponentially better lower bound than
the Linial and Shraibman family of methods.

On the positive side, this is very interesting
for quantum information, since (by Theorem~\ref{thm:lp-bell}),
it tells us that the set of quantum distributions
cannot be much larger than the local polytope, for any number
of inputs and outcomes.  For binary correlations, this
follows from the theorems of Tsirelson (Theorem~\ref{thm:tsirelson})
and Grothendieck (Proposition~\ref{prop:grothendieck}),
but no extensions are known for these results in
the more general setting.

The proof of Theorem~\ref{thm:nu-gamma2} proceeds by showing that an arbitrary quantum distribution may be written as an affine combination of quantum distributions over binary outcomes with uniform marginals. We can then conclude using Grothendieck's inequality. For the details of the proof, we will need
 two rather straightforward lemmas.
The first is a subadditivity-type property for $\cquasi$, and the second
allows us to extend the support of a distribution without affecting the
value of $\cquasi$.
\begin{lemma}\label{lem:composition}
 If $\mathbf{p}=\sum_{i\in[I]}q_i\mathbf{p}_i$, where $\mathbf{p}_i\in\C$ and $q_i\in\Real$ for all $i\in[I]$, then $\cquasi(\mathbf{p})\leq \sum_{i\in[I]}|q_i|\cquasi(\mathbf{p}_i)$.
\end{lemma}
\begin{proof}
 By definition, for each $\mathbf{p}_i$, there exists $\mathbf{p}_i^+,\mathbf{p}_i^-\in\L$ and $q_i^+,q_i^-\geq 0$ such that $\mathbf{p}_i=q_i^+\mathbf{p}_i^+-q_i^-\mathbf{p}_i^-$, and $q_i^++q_i^-=\cquasi(\mathbf{p}_i)$. Therefore, $\mathbf{p}=\sum_{i\in[I]}q_i(q_i^+\mathbf{p}_i^+-q_i^-\mathbf{p}_i^-)$ and $\sum_{i\in[I]}(|q_iq_i^+|+|q_iq_i^-|)=\sum_i|q_i|(q_i^++q_i^-)=\sum_i|q_i|\cquasi(\mathbf{p}_i)$.
\end{proof}

\begin{lemma}\label{lem:extension}
 Let $\mathbf{p},\mathbf{p}'\in\C$ be non-signaling distributions with inputs in
$\eks \times  \wye$ for both distributions, outcomes in $\A \times \B$ for $\mathbf{p}$, and outcomes in $\A' \times \B'$ for $\mathbf{p}'$, such that $\A\subseteq\A'$ and $\B\subseteq\B'$. If, for any $(a,b)\in\A\times\B$ $p'(a,b|x,y)=p(a,b|x,y)$, then $\cquasi(\mathbf{p}')=\cquasi(\mathbf{p})$.
\end{lemma}
\begin{proof}
Let $\ee=(\A'\times\B')\setminus(\A\times\B)$.
First, note that since $p'(a,b|x,y)=p(a,b|x,y)$ for any $(a,b)\in\A\times\B$, we have, by normalization of $\mathbf{p}$, $p'(a,b|x,y)=0$ for any $(a,b)\in\ee$.

[$\cquasi(\mathbf{p}')\leq\cquasi(\mathbf{p})$] Let $\mathbf{p}=q_+\mathbf{p}^+-q_-\mathbf{p}^-$ be an affine model for $\mathbf{p}$.
Obviously, this implies an affine model for $\mathbf{p}'$ by extending the local distributions $\mathbf{p}^+,\mathbf{p} ^-$ from $\A \times \B$ to $\A' \times \B'$, by setting $p^+(a,b|x,y)=p^-(a,b|x,y)=0$ for any $(a,b)\in\ee$, so $\cquasi(\mathbf{p}')\leq\cquasi(\mathbf{p})$.

[$\cquasi(\mathbf{p}')\geq\cquasi(\mathbf{p})$] Let $\mathbf{p}'=q_+\mathbf{p}'^+-q_-\mathbf{p}'^-$ be an affine model for $\mathbf{p}'$. We may not immediately derive an affine model for $\mathbf{p}$ since it could be the case that $p'^+(a,b|x,y)$ or $p'^-(a,b|x,y)$ is non zero for some $(a,b)\in\ee$. However, we have $q_+p'^+(a,b|x,y)-q_-p'^-(a,b|x,y)=p'(a,b|x,y)=0$ for any $(a,b)\in\ee$, so we may define an affine model $\mathbf{p}=q_+\mathbf{p}^+-q_-\mathbf{p}^-$, where $\mathbf{p}^+$ and $\mathbf{p}^-$ are distributions on $\A\times\B$ such that
$$
p^+(a,b|x,y)=p'^+(a,b|x,y)
+\frac{1}{A}\sum_{a'\notin\A}p'^+(a',b|x,y)
+\frac{1}{B}\sum_{b'\notin\B}p'^+(a,b'|x,y)
+\frac{1}{AB}\sum_{a'\notin\A,b'\notin\B}p'^+(a',b'|x,y),
$$
and similarly for $\mathbf{p}^-$. These are local since it suffices for Alice and Bob to use the local protocol for $\mathbf{p}'^+$ or  $\mathbf{p}'^-$ and for Alice to replace any output $a\notin\A$ by a uniformly random output $a'\in\A$ (similarly for Bob).
Therefore, we also have $\cquasi(\mathbf{p}')\geq\cquasi(\mathbf{p})$.
\end{proof}

Before proving Theorem~\ref{thm:nu-gamma2}, we first consider the special case of quantum distributions, for which $\qquasi(\mathbf{p})=1$.
As we shall see in Section~\ref{sec:smp}, this special case
implies the constant upper bound of Shi and Zhu on
approximating any quantum distribution~\cite{shi05}, which they
prove using diamond norms.  This also immediately gives an
upper bound on maximum Bell inequality violations for quantum
distributions, by
Theorem~\ref{thm:lp-bell},
which may be of independent interest in quantum information theory.

\begin{proposition}\label{prop:nu-quantum-dist}
 For any quantum distribution $\mathbf{p}\in\Q$, with inputs in
$\eks \times  \wye$ and outcomes in $\A \times \B$ with $A=|\A|, B=|\B|$,
\begin{enumerate}
 \item $\cquasi(\mathbf{p})\leq 2K_G+1$ when $A=B=2$,
 \item $\cquasi(\mathbf{p})\leq2AB(K_G+1)-1$ for any $A,B$.
\end{enumerate}
\end{proposition}
\begin{proof}
 \begin{enumerate}
  \item Since $A=B=2$, we may write the distribution as correlations and marginals, $\mathbf{p}=(C,M_A,M_B)$. Since $(C,M_A,M_B)\in\Q$, we also have $(C,0,0)\in\Q$, and by Tsirelson's theorem, $(C/K_G,0,0)\in\L$. Moreover, it is immediate that $(M_AM_B,M_A,M_B),(M_AM_B,0,0)$ and $(0,0,0)$ are local distributions as well, so that we have the following affine model for $(C,M_A,M_B)$
$$
(C,M_A,M_B)=K_G(C/K_G,0,0)+(M_AM_B,M_A,M_B)-(M_AM_B,0,0)-(K_G-1)(0,0,0).
$$
This implies that $\cquasi(C,M_A,M_B)\leq 2K_G+1$.

\item For the general case, we will reduce to the binary case. Let us introduce an additional output $\varnothing$, and set $\A'=\A\cup\{\varnothing\}$ and $\B'=\B\cup\{\varnothing\}$.
We first extend the distribution $\mathbf{p}$ to a distribution $\mathbf{p}'$ on $\A'\times\B'$ by setting $p'(a,b|x,y)=p(a,b|x,y)$ for any $(a,b)\in\A\times\B$, and $p'(a,b|x,y)=0$ otherwise. By Lemma~\ref{lem:extension}, we have $\cquasi(\mathbf{p})=\cquasi(\mathbf{p}')$.

For each $(\alpha,\beta)\in\A\times\B$, we also define a probability distribution $\mathbf{p}_{\alpha\beta}$ on $\A'\times\B'$:
$$
p_{\alpha\beta}(a,b|x,y)=
\begin{cases}
 p(\alpha , \beta|x,y) & \textrm{if } (a,b)=(\alpha,\beta),\\
 p(\alpha|x)-p(\alpha,\beta|x,y) & \textrm{if } (a,b)=(\alpha, \varnothing),\\
 p(\beta|y) -p(\alpha,\beta|x,y) & \textrm{if } (a,b)=(\varnothing, \beta),\\
1-p(\alpha|x)-p(\beta|y)+p(\alpha, \beta|x,y) & \textrm{if } (a,b)=(\varnothing,\varnothing),\\
 0 & \textrm{otherwise}.
\end{cases}
$$
Notice that $p_{\alpha \beta} \in \Q$, since a protocol for $p_{\alpha \beta}$
can be obtained from a protocol for $p$: Alice outputs $\varnothing$ whenever
her outcome is not $\alpha$, similarly for Bob.
Let $\A_\alpha=\{\alpha,\varnothing\}$ and $\B_\beta=\{\beta,\varnothing\}$. Since $p_{\alpha\beta}(a,b|x,y)=0$ when $(a,b)\notin\A_\alpha\times\B_\beta$, we may define distributions $\mathbf{p}'_{\alpha\beta}$ on $\A_\alpha\times\B_\beta$ such that $p_{\alpha\beta}'(a,b|x,y)=p_{\alpha\beta}(a,b|x,y)$ for all $(a,b)\in\A_\alpha\times\B_\beta$. By Lemma~\ref{lem:extension}, these are such that $\cquasi(\mathbf{p}'_{\alpha\beta})=\cquasi(\mathbf{p}_{\alpha\beta})$, and since these are binary distributions, $\cquasi(\mathbf{p}'_{\alpha\beta})\leq 2K_G+1$.
Let us define three distributions
$\mathbf{p_A}, \mathbf{p_B},\mathbf{p}_\varnothing$ on $\A'\times\B'$ as follows.
We let $\mathbf{p_A}(a,\varnothing|x,y)=p(a|x),
	\mathbf{p_B}(\varnothing,b|x,y)=p(b|y)$, and 0 everywhere else; and
$p_\varnothing(a,b|x,y)=1$ if $(a,b)=(\varnothing,\varnothing)$, and $0$ otherwise. These are product distributions, so
$\mathbf{p_A},\mathbf{p_B},\mathbf{p}_\varnothing\in\L$ and
$\cquasi=1$ for all three distributions.

We may now build the following affine model for $\mathbf{p}'$
$$
\mathbf{p}'=
	\sum_{(\alpha,\beta)\in\A\times\B} \mathbf{p}'_{\alpha\beta}
	-(B{-}1)\mathbf{p_A}- (A{-}1)\mathbf{p_B}
	-(AB{-}A{-}B{+}1)\mathbf{p}_\varnothing.
$$
From Lemma~\ref{lem:composition}, we conclude that $\cquasi(\mathbf{p}')\leq AB(2K_G+2)-1$.
\end{enumerate}
\end{proof}

The proof of Theorem~\ref{thm:nu-gamma2} immediately follows.
\begin{proof}[Proof of Theorem~\ref{thm:nu-gamma2}]
 By definition of $\qquasi(\mathbf{p})$, there exists $\mathbf{p}^+,\mathbf{p}^-\in\Q$ and $q_+,q_-\geq 0$ such that $\mathbf{p}=q_+\mathbf{p}^+-q_-\mathbf{p}^-$ and $q_++q_-=\qquasi(\mathbf{p})$. From Lemma~\ref{lem:composition}, $\cquasi(\mathbf{p})\leq q_+\cquasi(\mathbf{p}^+)+q_- \cquasi(\mathbf{p}^-)$, and Proposition~\ref{prop:nu-quantum-dist} immediately concludes the proof.
\end{proof}

\section{Upper bounds for non-signaling distributions} \label{sec:smp}

We have seen that  if a distribution can be simulated
using $t$ bits of communication, then it may be represented by an
affine model with coefficients exponential in $t$ (Lemma~\ref{lm:decreased-correlations-pi}).
In this section, we consider the converse: how much communication
is sufficient to simulate a distribution, given an affine model?
This approach allows us to show that any (shared randomness or entanglement-assisted)
communication protocol can be simulated with simultaneous messages, with
an exponential cost to the simulation, which was previously known only
in the case of Boolean functions~\cite{yaofinger03,shi05,gkr06}.
Our results imply for example that for
any quantum distribution $\mathbf{p}\in \Q$,
$Q_\eps^\parallel(\mathbf{p})= O(\log(n))$, where $n$ is the input size.
This in effect replaces arbitrary entanglement
in the state being measured,
with logarithmic quantum communication (using no additional
resources such as shared randomness).
We use the
superscript $\parallel$ to indicate the simultaneous messages model, where
Alice and Bob each send a message to the referee, who without knowing
the inputs, outputs the value of the function, or more generally,
outputs $a,b$ with the correct probability distribution conditioned
on the inputs $x,y$.

\begin{thm}\label{thm:smp}
For any distribution $\mathbf{p}\in \C$ with inputs in
$\eks \times  \wye$ with $|\eks \times \wye| \leq 2^n$,
and outcomes in $\A \times \B$ with $A=|\A|, B=|\B|$,
and any $\epsilon, \delta < 1/2$,
\begin{enumerate}
\item $R_{\epsilon + \delta}^{\parallel,\pub}(\mathbf{p}) \leq
	16 \left[\frac{AB\cquasi^\epsilon(\mathbf{p})}{\delta}\right]^2
	\ln\left[\frac{4AB}{\delta}\right] \log(AB)$,
\item $Q_{\epsilon+\delta}^{\parallel}(\mathbf{p}) \leq
O\left((AB)^5\left[\frac{\cquasi^\epsilon(\mathbf{p})}{\delta}\right]^4\ln\left[\frac{AB}{\delta}\right]\log(n)\right)$.
\end{enumerate}
\end{thm}

The general idea of the proof is to build a communication protocol for $\mathbf{p}$ based on an affine combination $\mathbf{p}=q_+\mathbf{p}^+-q_-\mathbf{p}^-$, where $\mathbf{p}^+$ and $\mathbf{p}^-$ are local (or quantum) distributions. By sending sufficiently many samples of $\mathbf{p}^+$ and $\mathbf{p}^-$ to the referee (which does not require any communication between Alice and Bob), the referee can estimate these distributions and therefore simulate their affine combination $\mathbf{p}$. To quantify the number of samples that are necessary to achieve some precision, we use
Hoeffding's inequality~\cite{mcdiarmid}.
\begin{proposition}[Hoeffding's inequality]
\label{prop:hoeffding}
Let $X$ be a random variable with values in $[a,b]$. Let $X_t$ be the $t$-th of
$T$ independent trials of $X$, and $S=\frac{1}{T}\sum_{t=1}^T X_t$.

Then,
$\Pr[S-E(X) \geq \beta]\leq e^{-\frac{2T\beta^2}{(b-a)^2}}$,
and $\Pr[E(X)-S \geq \beta]\leq e^{-\frac{2T\beta^2}{(b-a)^2}}$, for any $\beta\geq0$.
\end{proposition}

We will also use the following lemma.
\begin{lemma}\label{lem:estimated-distribution}
 Let $\mathbf{p}$ be a probability distribution on $\V$ with $V=|\V|$, and $e:\Real^+\to\Real^+$.
For each $v\in \V$, let $Q_v$ be a random variable such that
$\forall \beta\geq 0$,
$\Pr[{Q}_v\geq p(v) + \beta ]\leq e(\beta)$
and $\Pr[{Q}_v \leq p(v) - \beta ]\leq e(\beta)$.

Then, given samples $\{Q_v:v\in \V\}$, and without knowing $\mathbf{p}$, we may simulate a probability distribution $\mathbf{p'}$ such that $\delta(\mathbf{p'},\mathbf{p})\leq 2V[\beta+e(\beta)]$.
\end{lemma}

\begin{proof}
 In order to use the variables $Q_v$ as estimations for $p(v)$, we must first make them positive, and then renormalize them so that they sum up to $1$.
Let $R_v = \max\{0,Q_v\}$.
Then we may easily verify that
\begin{eqnarray*}
 \Pr[R_v \geq p(v) +  \beta ]&\leq& e(\beta),\\
\Pr[R_v \leq p(v) - \beta ]&\leq& e(\beta).
\end{eqnarray*}
%

For any subset $\ee\subseteq\V$ of size $E=|\ee|$, we also define the estimates $R_{\ee}=\sum_{v\in \ee} R_v$ for $p(\ee)$. For any $v$, we have $R_{v}-p(v) \geq \beta$ with probability at least $1-e(\beta)$. Therefore, with probability at least $1-E e(\beta)$, we have $R_{v}-p(v) \geq \beta$ simultaneously for all $v\in \ee$, and therefore by summation also $R_\ee-p(\ee) \geq E\beta$. Similarly, with probability at least $1-E e(\beta)$, we have $p(v)-R_{v} \geq \beta$ simultaneously for all $v\in \ee$, and therefore also $p(\ee)-R_\ee \geq E\beta$. Hence, we have the following bounds for $R_\ee$
For any subset $\ee\subseteq\V$ of size $E=|\ee|$, we also define the estimates $R_{\ee}=\sum_{v\in \ee} R_v$ for $p(\ee)$. By summing,
\begin{eqnarray*}
 \Pr[R_\ee \geq p(\ee) + E\beta ]&\leq& Ee(\beta),\\
\Pr[R_\ee \leq p(\ee) - E\beta ]&\leq& Ee(\beta).
\end{eqnarray*}

In order to renormalize the estimated probabilities, let $R_{\V}=\sum_{v\in{\V}} R_{v}$. If $R_{\V}>1$, we use as final estimates $S_{v}=R_{v}/R_{\V}$. On the other hand, if $R_{\V}\leq 1$, we keep $S_{v}=R_{v}$ and introduce a dummy output $\varnothing\notin\V$ with estimated probability $S_\varnothing=1-R_{\V}$ (we extend the original distribution to $\V\cup\{\varnothing\}$, setting $p(\varnothing)=0$). By outputting $v$ with probability $S_{v}$, we then simulate some distribution $p'(v)=E(S_{v})$, and it suffices to show that $|E(S_{\ee})-p(\ee)|\leq2V[\beta+e(\beta)]$ for any $\ee\subseteq\V\cup\{\varnothing\}$.

We first upper bound $E(S_\ee)$ for $\ee\in\V$. Since $S_\ee\leq R_\ee$, we obtain from the bounds on $R_\ee$ that $\Pr[S_\ee\geq p(\ee) + E\beta ]\leq Ee(\beta)$. Therefore, we have $S_\ee<p(\ee)+E\beta$ with probability at least $1-Ee(\beta)$, and $S_\ee\leq 1$ with probability at most $Ee(\beta)$. This implies that
$E(S_\ee)\leq p(\ee)+E\left[\beta+e(\beta)\right]$.

To lower bound $E(S_\ee)$, we note that with probability at least $1-Ee(\beta)$, we have $R_\ee>p(\ee)-E\beta$, and with probability at least $1-Ve(\beta)$, we have $R_{\V}<1+V\beta$. Therefore, with probability at least $1-(E+V)e(\beta)$, both these events happen at the same time, so that $S_\ee=R_\ee/R_{\V}>(p(\ee)-E\beta)(1-V\beta)\geq p(\ee)-(E+V)\beta$. This implies that
$E(S_\ee)\geq p(\ee)-(E+V)\left[\beta+e(\beta)\right]$.
Since $S_\varnothing=1-S_{\V}$, this also implies that
$E(S_\varnothing)\leq 2V\left[\beta+e(\beta)\right]$.
\end{proof}

\begin{proof}[Proof of Theorem~\ref{thm:smp}]
\noindent
{ 1.} Let $\Lambda=\cquasi(\mathbf{p})$, $\mathbf{p}= q_+ \mathbf{p}^+ - q_- \mathbf{p}^-$,
with $q_+, q_-\geq 0$, $q_+ + q_- = \Lambda$ and
$\mathbf{p}^+, \mathbf{p}^- \in \L$.
Let $P^+,P^-$ be protocols for $\mathbf{p}^+$ and $\mathbf{p}^-$,
respectively.  These protocols use shared randomness but no
communication.

To simulate $\mathbf{p}$, Alice and Bob make $T$ independent
runs of $P^+$, where we label the outcome of the $t$-th run
$(a_t^+,b_t^+)$.  Similarly, let $(a_t^-,b_t^-)$ be the
outcome of the $t$-th run of $P^-$.  They send the list
of outcomes to the referee.

The idea is for the referee to estimate $p(a,b|x,y)$
based on the $2T$ samples,
and output according to the estimated distribution.
Let $P^+_{t,a,b}$ be an indicator variable which equals 1 if
$a_t^+=a$ and $b_t^+=b$, and 0 otherwise.
Define $P^-_{t,a,b}$ similarly. Furthermore, let $P_{t,a,b}=q_+P^+_{t,a,b}-q_-P^-_{t,a,b}$.
Then $E(P_{t,a,b})=p(a,b|x,y)$ and $P_{t,a,b}\in [-q_-,q_+]$.

Let $P_{a,b}= \frac{1}{T}\sum_{t=1}^{T}P_{t,a,b}$ be the referee's estimate
for $p(a,b|x,y)$. By Hoeffding's inequality,
\begin{eqnarray*}
 \Pr[P_{a,b}\geq p(a,b|x,y) + \beta ]&\leq& e^{-\frac{2T\beta^2}{\Lambda^2}},\\
\Pr[P_{a,b} \leq p(a,b|x,y)-\beta ]&\leq& e^{-\frac{2T\beta^2}{\Lambda^2}}.
\end{eqnarray*}

Lemma~\ref{lem:estimated-distribution} with $\V=\A\times \B$, ${Q}_{a,b}=P_{a,b}$ and $e(\beta)=e^{-\frac{2T\beta^2}{\Lambda^2}}$ then implies that the referee may simulate a probability distribution $\mathbf{p'}$ such that
$\delta(\mathbf{p'},\mathbf{p})\leq 2AB(\beta+e^{-\frac{2T\beta^2}{\Lambda^2}})$.
It then suffices to set $\beta=\frac{\delta}{4AB}$, and $T=8 \left[\frac{AB\Lambda}{\delta}\right]^2\ln\left[\frac{4AB}{\delta}\right]$ to conclude the proof, since Alice sends $2T\log A$ and Bob sends $2T\log B$ bits to the referee.

For $\cquasi^\epsilon$, apply this proof to the distribution
$\mathbf{p''}$ with statistical distance $\delta(\mathbf{p},\mathbf{p''})\leq \epsilon$
and $\cquasi(\mathbf{p''})=\cquasi^\epsilon(\mathbf{p})$.

Note that the same proof gives an upper bound on  $R_{\epsilon + \delta}^{\parallel,\ent}$ in terms of $\qquasi$.

\noindent
{ 2.}
If shared randomness is not available but quantum
messages are, then we can use  quantum fingerprinting~\cite{bcwdw01,yaofinger03} to send the
results of the repeated protocol to the referee.
Let $(a^+(r),b^+(r))$ be the outcomes of $P^+$
using $r$ as shared randomness.
We use the random variable $A^+_{a}(r)$ as an indicator variable for
$a^+(r)=a$; similarly $B^+_{b}$, and $P^+_{\ee}=\sum_{(a,b)\in\ee}A^+_{a}B^+_{b}$.

We can easily adapt the proof of Newman's Theorem~\cite{newman91},
to show that there exists a set of $L$ random strings
${\cal{R}} = \{r_1,\ldots r_L\}$
such that
$\forall x,y,
\abs{E_{r_i\in {\cal R}}(\tilde{P}_\ee^+(r_i))
	- {E(P_\ee^+)}  }
\leq \alpha $ provided $L \geq \frac{4n}{\alpha^2}$,
where $n$ is the input length, and $\tilde{P}_\ee^+$ is the random variable
where randomness is taken from $\cal R$. In other words, by taking the randomness from $\cal R$, we may simulate a probability distribution $\tilde{\mathbf{p}}^+$ such that $\delta(\tilde{\mathbf{p}}^+,\mathbf{p}^+)\leq\alpha$.

For each $a,b\in \A\times \B$, Alice and Bob send $T$ copies of the states
$\ket{\phi_a^+}=\frac{1}{\sqrt{L}}
	\sum_{1\leq i\leq L}
	   \ket{A^+_a(r_i)}\ket{1}\ket{i}$ and
$\ket{\phi_b^+} =\frac{1}{\sqrt{L}}
	\sum_{1\leq i\leq L}
	    \ket{1}\ket{B^+_a(r_i)}\ket{i}$ to the referee.
The inner product is
$$\braket{\phi_a^+}{\phi_b^+}=\frac{1}{{L}}
	\sum_{1\leq i\leq L} \braket{A^+_a(r_i)}{1}\braket{1}{B^+_b(r_i)}
		= \tilde{p}^+(a,b|x,y),$$
where the expectation is taken over the random choices $r_1,\ldots r_L$.

The referee then uses inner product estimation~\cite{bcwdw01}:
for each copy, he performs a measurement on $\ket{\phi_a^+}\otimes\ket{\phi_b^+}$ to obtain a random variable $Z^+_{t,a,b}\in\{0,1\}$ such that $\Pr[Z^+_{t,a,b}=1]=\frac{1-\abs{\braket{\phi_b^+}{\phi_a^+}}^2}{2}$, then he sets $Z^+_{a,b}=\frac{1}{T}\sum_{t=1}^TZ^+_{t,a,b}$. Let ${Q}^+_{a,b}=\sqrt{1-2Z^+_{a,b}}$ if $Z^+_{a,b}\leq 1/2$ and ${Q}^+_{a,b}=0$ otherwise. This serves as an approximation for $\tilde{p}^+(a,b|x,y)=\abs{\braket{\phi_b^+}{\phi_a^+}}$, and Hoeffding's inequality then yields
\begin{eqnarray*}
 \Pr[Q^+_{a,b} \geq  \tilde{p}^+(a,b|x,y) +  \beta ]&\leq& e^{-\frac{T\beta^4}{2}},\\
\Pr[ Q^+_{a,b} \leq \tilde{p}^+(a,b|x,y)-\beta ]&\leq& e^{-\frac{T\beta^4}{2}}.
\end{eqnarray*}

Let $Q^-_{a,b}$ be an estimate for $\tilde{p}^-(a,b|x,y)$ obtained using the same method. The referee then obtains an estimate for $\tilde{p}(a,b|x,y)=q_+\tilde{p}^+(a,b|x,y)-q_-\tilde{p}^-(a,b|x,y)$, by setting $Q_{a,b}=q_+Q^+_{a,b}+q_-Q^-_{a,b}$, such that
\begin{eqnarray*}
 \Pr[Q_{a,b}\geq \tilde{p}(a,b|x,y) + \beta ]&\leq& 2e^{-\frac{T\beta^4}{2\Lambda^4}},\\
\Pr[ Q_{a,b} \leq \tilde{p}(a,b|x,y)- \beta ]&\leq& 2e^{-\frac{T\beta^4}{2\Lambda^4}}.
\end{eqnarray*}

Lemma~\ref{lem:estimated-distribution} with $e(\beta)=2e^{-\frac{T\beta^4}{2\Lambda^4}}$ then implies that the referee may simulate a probability distribution $\mathbf{p}^s$ such that
$\delta(\mathbf{p}^s,\tilde{\mathbf{p}})\leq 2AB(\beta+2e^{-\frac{T\beta^4}{2\Lambda^4}})$.
Since $\delta(\tilde{\mathbf{p}},\mathbf{p})\leq\Lambda\alpha$,
we need to pick $T,L=\frac{4n}{\alpha}$ large enough so that
$\Lambda\alpha+2AB\left[\beta+2e^{-T\beta^4/2\Lambda^4}\right]\leq \delta$.
Setting $\alpha=\frac{\delta}{2\Lambda}$,
$\beta=\frac{\delta}{8AB}$,
$T=2\frac{\Lambda^4}{\beta^4}\ln(\frac{16AB}{\delta})
	=  2^{13}\left[\frac{AB\Lambda}{\delta}\right]^4\ln(\frac{16AB}{\delta})$ and
$L=\frac{4n}{\alpha^2}=  \frac{16n\Lambda^2}{\delta^2}$,
the total complexity of the protocol is
$4ABT(\log(L)+2) = O((AB)^5\left[\frac{\Lambda}{\delta}\right]^4\ln\left[\frac{AB}{\delta}\right]\log(n))$,
(we may assume that $\frac{\Lambda}{\delta}\leq n^{1/4}$, otherwise
this protocol performs worse than the trivial protocol).
\end{proof}

In the case of Boolean functions, corresponding to correlations $C_f(x,y)\in\{\pm 1\}$ (see Def.~\ref{def:boolean-functions}), the referee's job is made easier
by the fact that he only needs to determine the sign of the correlation
with probability $1-\delta$. This allows us to get some improvements
in the upper bounds.  Similar improvements can be
obtained for other types of promises on the distribution.
\begin{thm}
\label{thm:smp-boolean}
 Let $f:\{0,1\}^n \times \{0,1\}^n \rightarrow \{0,1\}$, with associated sign matrix $C_f$, and $\epsilon, \delta < 1/2$.
\begin{enumerate}
\item
$R_{\delta}^{\parallel,\pub}(f) \leq 4\left[\frac{\cquasi^\epsilon(C_f)}{1-2\epsilon}\right]^2\ln(\frac{1}{\delta})
$,
\item
$Q_{\delta}^{\parallel}(f)
\leq O\left(\log(n) \left[\frac{\cquasi^\epsilon(C_f)}{1-2\epsilon}\right]^4\ln(\frac{1}{\delta})\right)
$.
\end{enumerate}
\end{thm}

From Lemmas~\ref{lem:epsilon-alpha} and~\ref{lem:gamma2-infinity}, these bounds may also be expressed in terms of $\qls^\alpha$, and the best upper bounds are obtained from $\qls^\infty(C_f)=\frac{1}{\epsilon^\ent(C_f)}$. The first item then coincides with the upper bound of~\cite{ls07}.

Together with the bound between $\cquasi$ and $\qquasi$ from Section~\ref{sec:gamma-vs-nu}, and the lower bounds on communication complexity from Section~\ref{sec:lower-bounds}, Theorems~\ref{thm:smp} and~\ref{thm:smp-boolean} immediately imply the following corollaries.

\begin{cor}
\label{cor:smp}
Let $f:\{0,1\}^n \times \{0,1\}^n \rightarrow \{0,1\}$.
For any $\epsilon, \delta < 1/2$, if $Q_\epsilon^\ent(f) \leq q$, then
\begin{enumerate}
\item $R_{\delta}^{\parallel,\pub}(f) \leq K_G^2 \cdot 2^{2q+2} \ln(\frac{1}{\delta})\frac{1}{(1{-}2\epsilon)^2} $,
\item $Q_{\delta}^{\parallel}(f)\leq O\left(\log(n) 2^{4q}\ln(\frac{1}{\delta})\frac{1}{(1{-}2\epsilon)^4}\right)$.
\end{enumerate}
Let $\mathbf{p}\in \C$ be a distribution  with inputs in
$\eks \times  \wye$ with $|\eks \times \wye| \leq 2^n$, and outcomes in $\A \times \B$ with $A=|\A|, B=|\B|$. For any $\epsilon, \delta < 1/2$, if $Q_\epsilon^{\ent}(\mathbf{p}) \leq q$, then
\begin{enumerate}
\addtocounter{enumi}{2}
\item $R_{\epsilon + \delta}^{\parallel,\pub}(\mathbf{p}) \leq
	O\left(2^{4q}\frac{(AB)^4}{\delta^2}
	\ln^2\left[\frac{AB}{\delta}\right] \right)$,
\item $Q_{\epsilon+\delta}^{\parallel}(\mathbf{p}) \leq
O\left(2^{8q}\ \frac{(AB)^9}{\delta^4}\ln\left[\frac{AB}{\delta}\right]\log(n)\right)$.
\end{enumerate}
\end{cor}

The first two items can be compared to results of Yao, Shi and Zhu, and
Gavinsky \textit{et al.}~\cite{yaofinger03,shi05,gkr06},
who show how to simulate any (logarithmic) communication protocol for Boolean functions
in the simultaneous messages
model, with an exponential blowup in communication.
The last two items extend these results to arbitrary non-signaling distributions.

In particular, Item~3 gives in the special case $q=0$, that is, $\mathbf{p}\in \Q$,
a much simpler proof of the constant upper bound on approximating
quantum distributions, which Shi and Zhu prove using
sophisticated techniques based on diamond norms~\cite{shi05}. Moreover, Item~3 is
much more general as it also allows to simulate protocols requiring quantum communication in addition to entanglement.
As for Item~4, it also has new interesting consequences.
For example, it implies that quantum distributions ($q=0$)
can be approximated with logarithmic quantum communication
in the simultaneous messages model, using no additional resources such as shared randomness,
and regardless of the amount of entanglement
in the bipartite state measured by the two parties.

\section{Conclusion and open problems}

By studying communication complexity in the framework
provided by the study of quantum non-locality (and beyond),
we have given very natural and intuitive interpretations of
the otherwise very abstract lower bounds of Linial and
Shraibman.  Conversely, bridging this gap has allowed us to
port these very strong and mathematically elegant
lower bound methods to the much more general problem
of simulating non-signaling distributions.

Since many communication problems may be reduced to the task
of simulating a non-signaling distribution, we hope to see
applications of this lower bound method to concrete problems
for which standard techniques do not apply, in particular for
cases that are not Boolean functions, such as non-Boolean
functions, partial functions or relations. Let us also note
that our method can be generalized to multipartite non-signaling
distributions, and will hopefully lead to applications in
the number-on-the-forehead model, for which quantum
lower bounds seem hard to prove.

In the case of binary distributions with uniform marginals
(which includes in particular Boolean functions), Tsirelson's
theorem (Theorem~\ref{thm:tsirelson}) and
the existence of Grothendieck's constant (Proposition~\ref{prop:grothendieck})
imply that there is at most a constant gap between $\cls$ and $\qls$.
For this reason, it was known that Linial and Shraibman's
factorization norm lower bound technique give lower bounds
of the same of order for classical and quantum communication
(note that this is also true for the related discrepancy method).
Despite the fact that Tsirelson's theorem and Grothendieck's inequality
are not known to extend beyond the case of Boolean outcomes with
uniform marginals, we
have shown that in the general case of distributions, there
is also a constant gap between $\cquasi$ and $\qquasi$.
While this may be seen as a negative result, this also reveals
interesting information about the structure of the sets of local
and quantum distributions. In particular, this
could have interesting consequences for the study of non-local games.


\section*{Acknowledgements}
We are grateful to Benjamin Toner for pointing us towards the existing
literature on non-signaling distributions as well as very useful discussions
of the Linial and Shraibman lower bound on communication complexity.
We also thank
Peter H{\o}yer,
Troy Lee,
Oded Regev,
Mario Szegedy,
and Dieter van Melkebeek
with whom we had many stimulating discussions.
Part of this work was done while J. Roland was
affiliated with FNRS Belgium and U.C. Berkeley.
The research was supported by the EU 7th framework program QCS, and ANR D\'efis QRAC.

\bibliography{biblioCausal}

\newpage
\appendix

\section{Proof of Lemma~\ref{lemma:quantum-communication}\label{appendix:quantum-communication}}

We give the proof of  Lemma~\ref{lemma:quantum-communication}, which relates
the outcome of communication protocols to vectors of bounded norm.
\setcounter{lemma}{2}
\begin{lemma}[\cite{kremer, yao, ls07}]
Let $(C, M_A, M_B)$ be a distribution simulated by a
quantum protocol with shared entanglement
using $q_A$ qubits of communication from Alice to Bob and $q_B$ qubits from Bob to Alice.
There exist vectors $\vec{a}(x),\vec{b}(y)$ with
$\norm{\vec{a}(x)}\leq 2^{q_B}$ and $\norm{\vec{b}(y)} \leq 2^{q_A}$ such that
$C(x,y) = \vec{a}(x)\cdot\vec{b}(y)$.
\end{lemma}

The proof relies on the following observation:
\begin{claim}\label{claim:psit}
Let $\ket{\psi_t}$ be the entangled state shared by Alice and Bob after the first $t=t_A+t_B$ qubits of communication ($t_A$ bits from Alice to Bob, and $t_B$ bits from Bob to Alice). This state may be written as $\ket{\psi_t}=\sum_{i\in I}\mu_i\sum_{T\in\{0,1\}^t} A_T\ket{\alpha^{(i)}}B_T\ket{\beta^{(i)}}$, where $\sum_i |\mu_i|^2=1$, $\{\ket{\alpha^{(i)}}:\forall i\in I\}$ and $\{\ket{\beta^{(i)}:\forall i\in I}\}$ are orthonormal bases for Alice and Bob's initial registers respectively and $A_T,B_T$ are linear operators such that:
\begin{myitemize}
\item $A_0$,$B_0$ are the identity operators on Alice and Bob's initial registers, respectively,
\item $A_T$ are linear operators acting on Alice's initial register and depending on her input only, satisfying
$$\sum_{T\in\{0,1\}^t}\norm{A_T\ket{\psi_A}}^2=2^{t_B}$$ for all (unit) state $\ket{\psi_A}$ of Alice's register.
\item $B_T$ are linear operators depending on Bob's input only, satisfying $\sum_{T\in\{0,1\}^t}\norm{B_T\ket{\psi_B}}^2=2^{t_A}$ for all (unit) state $\ket{\psi_B}$ of Bob's register.
\end{myitemize}
\end{claim}

\begin{proof}[Proof of Claim~\ref{claim:psit}]
We prove this by induction over $t$. This is true for $t=0$, since using Schmidt decomposition, we may write the initial entangled state shared by Alice and Bob, before the quantum communication protocol is initiated, as $\ket{\psi_0}=\sum_{i\in I}\mu_i \ket{\alpha^{(i)}}\ket{\beta^{(i)}}$, where $\sum_i |\mu_i|^2=1$ and $\{\ket{\alpha^{(i)}}:\forall i\in I\}$ and $\{\ket{\beta^{(i)}:\forall i\in I}\}$ are orthonormal bases for Alice and Bob's registers respectively (as is, these are actually just orthonormal, but we can always obtain a basis by setting $\mu_i=0$ for the missing basis vectors).

If this is true for $t-1$, then we have $\ket{\psi_{t-1}}=\sum_{i\in I}\mu_i\sum_{T\in\{0,1\}^{t-1}} A_T\ket{\alpha^{(i)}}B_T\ket{\beta^{(i)}}$, where\linebreak $\sum_{T\in\{0,1\}^{t-1}}\norm{A_T\ket{\alpha^{(i)}}}^2=2^{t_B}$ and $\sum_{T\in\{0,1\}^{t-1}}\norm{B_T\ket{\beta^{(i)}}}^2=2^{t_A-1}$ for all $i\in I$ (we assume without loss of generality that the $t$'s qubit is sent by Alice to Bob). Alice's operation at turn $t$ will be to apply some unitary operation $U_t$ on her register, then send one of the qubits in her register to Bob. By isolating this qubit, we define the linear operators $A_{T0}$ and $A_{T1}$ to be such that $U_tA_T\ket{\alpha^{(i)}}=A_{T0}\ket{\alpha^{(i)}}\ket{0}+A_{T1}\ket{\alpha^{(i)}}\ket{1}$ for all $i\in I$. Unitarity then implies that
$\norm{A_{T0}\ket{\alpha^{(i)}}}^2+\norm{A_{T1}\ket{\alpha^{(i)}}}^2=\norm{A_T\ket{\alpha^{(i)}}}^2$, and as a consequence $\sum_{T\in\{0,1\}^t}\norm{A_T\ket{\alpha^{(i)}}}^2=2^{t_B}$. We then have
\begin{eqnarray}
\ket{\psi_t}&=&\sum_{i\in I}\mu_i\sum_{T\in\{0,1\}^{t-1}} \left[A_{T0}\ket{\alpha^{(i)}}\ket{0}B_T\ket{\beta^{(i)}}+A_{T1}\ket{\alpha^{(i)}}\ket{1}B_T\ket{\beta^{(i)}}\right]\\
&=&\sum_{i\in I}\mu_i\sum_{T\in\{0,1\}^{t}} A_T\ket{\alpha^{(i)}}B_T\ket{\beta^{(i)}},
\end{eqnarray}
where, for all $T\in\{0,1\}^{t-1}$, we have defined linear operators $B_{T0},B_{T1}$ such that $B_{T0}\ket{\beta^{(i)}}=\ket{0}B_T\ket{\beta^{(i)}}$ and $B_{T1}\ket{\beta^{(i)}}=\ket{1}B_T\ket{\beta^{(i)}}$ for all $i\in I$, considering that the additional qubit is in Bob's hands at the end of turn $t$. Furthermore, we have $\norm{B_{T0}\ket{\beta^{(i)}}}^2+\norm{B_{T1}\ket{\beta^{(i)}}}^2=2\norm{B_T\ket{\beta^{(i)}}}^2$, and as a consequence $\sum_{T\in\{0,1\}^t}\norm{B_T\ket{\beta^{(i)}}}^2=2^{t_A}$, which completes the proof of our claim.
\end{proof}

\begin{proof}[Proof of Lemma~\ref{lemma:quantum-communication}]
At the end of the quantum communication protocol, Alice and Bob share a quantum state $\ket{\psi_q}$ satisfying Claim~\ref{claim:psit} for $t=q$. Alice and Bob then perform binary ($\{+1,-1\}$-valued) measurements $A$ and $B$ on their respective parts of the state. By orthonormality of the states $\ket{\psi_q^{(i)}}$, we have for the correlation
\begin{eqnarray}
C&=&\bra{\psi_q}AB\ket{\psi_q}\\
&=&\sum_{i,j\in I}\mu_i^*\mu_j \sum_{T,U\in\{0,1\}^q} \bra{\alpha^{(i)}}A_T^\dagger A A_U\ket{\alpha^{(j)}} \bra{\beta^{(i)}}B_T^\dagger B B_U\ket{\beta^{(j)}}.
\end{eqnarray}
We may now define the vectors $\vec{a}(x)$ and $\vec{b}(y)$ in a $2^{2t}|I|^2$-dimensional complex vector space, with coordinates
\begin{eqnarray}
a_{TUij}(x)&=&\mu_i \bra{\alpha^{(j)}}A_U^\dagger A A_T \ket{\alpha^{(i)}},\\
b_{TUij}(x)&=&\mu_j \bra{\beta^{(i)}}B_T^\dagger B B_U \ket{\beta^{(j)}},\quad\forall\ T,U\in\{0,1\}^{q},i,j\in I,
\end{eqnarray}
so that $C=\vec{a}(x)\cdot\vec{b}(y)$. Moreover, using the fact that the $\ket{\alpha^{(j)}}$'s define an orthonormal basis for Alice's register and the property on the norms of the operators $A_T$, we have
\begin{eqnarray}
\norm{\vec{a}(x)}^2&=&\sum_{i,j\in I}|\mu_i|^2 \sum_{T,U\in\{0,1\}^q} |\bra{\alpha^{(j)}}A_U^\dagger A A_T \ket{\alpha^{(i)}}|^2\\
&=&\sum_{i\in I}|\mu_i|^2 \sum_{T,U\in\{0,1\}^q} \norm{A_U^\dagger A A_T \ket{\alpha^{(i)}}}^2\\
&\leq&\sum_{i\in I}|\mu_i|^2 \sum_{T,U\in\{0,1\}^q} \norm{A_U^\dagger\ket{\phi_T^{(i)}}}^2 \norm{A_T\ket{\alpha^{(i)}}}^2=2^{2q_B},
\end{eqnarray}
where $\ket{\phi_T^{(i)}}$ is the renormalized state $A A_T \ket{\alpha^{(i)}}$. So, we have $\norm{\vec{a}(x)}\leq2^{q_B}$, and similarly $\norm{\vec{b}(y)}\leq2^{q_A}$.
\end{proof}

\end{document}